\documentclass[journal,12pt,onecolumn,draftclsnofoot,]{IEEEtran}
\usepackage[pdftex]{graphicx}
\usepackage{epstopdf}
\DeclareGraphicsExtensions{.pdf,.jpeg,.png,.eps}
\usepackage{cite}
\usepackage{amssymb}
\usepackage{amsmath}
\usepackage{graphics}
\usepackage{graphicx}
\usepackage{float}
\usepackage[T1]{fontenc}
\usepackage[utf8]{inputenc}
\usepackage{authblk}
\usepackage{authblk}
\usepackage{tabularx}
\usepackage{balance}
\usepackage{ amssymb }
\usepackage{algorithm2e}
\usepackage{algorithmic}
\usepackage{mathtools}
\usepackage{gensymb}
\usepackage{subfig}
\usepackage{multicol}
\usepackage{lipsum}
\usepackage{chngpage}
\usepackage{calc}
\usepackage{ragged2e}
\usepackage{tabularx}
\usepackage[utf8]{inputenc}
\usepackage[english]{babel}
\usepackage{color}
\usepackage[normalem]{ulem}
\usepackage{amsmath, amsthm, amssymb}
\usepackage{epsfig}
\usepackage{latexsym}
\usepackage{graphicx}
\usepackage{cite}
\usepackage{epsfig}
\usepackage{hyperref}
\usepackage{soul}
\hypersetup{urlcolor=black, colorlinks=false}

\newtheorem{theorem}{Theorem}

\newtheorem{proposition}{Proposition}

\begin{document}

%\markboth{IEEE Photonics Technology Letters}{Marshoud \MakeLowercase{\textit{et al.}}: Non-Orthogonal Multiple Access for Visible Light Communications}

\title{On the Performance of Visible Light Communications Systems with Non-Orthogonal Multiple Access }

\author{
Hanaa~Marshoud,~\IEEEmembership{Student~Member,~IEEE},~Paschalis~C.~Sofotasios,~\IEEEmembership{Member,~IEEE}, Sami~Muhaidat,~\IEEEmembership{Senior~Member,~IEEE},~and~George~K.~Karagiannidis,~\IEEEmembership{Fellow,~IEEE}

%\thanks{Manuscript received xxx; revised xxx.}%
\thanks{H.~Marshoud is with the Department of Electrical and Computer Engineering, Khalifa University, PO Box 127788,
Abu Dhabi, UAE (email: hanaa.marshoud@kustar.ac.ae).}%
\thanks{P. C. Sofotasios is with the Department of Electrical and Computer Engineering, Khalifa
University, PO Box 127788, Abu Dhabi, UAE, and  with the Department of Electronics and Communications Engineering, Tampere University of Technology, 33101 Tampere, Finland (e-mail: p.sofotasios@ieee.org).}%
\thanks{S.~Muhaidat is with the Department of Electrical and Computer Engineering, Khalifa University, PO Box 127788, Abu Dhabi, UAE (email: sami.muhaidat@kustar.ac.ae), and with Institute for Communication Systems, University of Surrey, GU2 7XH, Guildford, United Kingdom (email: muhaidat@ieee.org).}%
\thanks{G.~K.~Karagiannidis is with  the Department of Electrical and Computer Engineering, Aristotle University of Thessaloniki, 54 124 Thessaloniki, Greece (email:  geokarag@auth.gr).}
%\thanks{Copyright (c) 2015 IEEE. Personal use of this material is permitted. However, permission to use this material for any other purposes must be obtained from the IEEE by sending a request to pubs-permissions@ieee.org.}%
%\thanks{Digital Object Identifier $\#\#\#$}%
}

\maketitle

\begin{abstract}
Visible light communications (VLC)  have been recently proposed as a promising and efficient solution to indoor ubiquitous broadband connectivity.
In this paper, non-orthogonal multiple access, which has been recently proposed as an effective  scheme for fifth generation ($5$G) wireless networks, is  considered in the context of VLC systems, under different channel uncertainty models. To this end, we first derive a novel closed-form expression for the bit-error-rate (BER) under perfect channel state information (CSI). Capitalizing on this, we quantify the effect of noisy and outdated CSI by deriving a simple approximated expression for the former and a tight upper bound for the latter. The offered results are corroborated by respective results from extensive Monte Carlo simulations and are used to provide useful  insights on the effect of imperfect CSI knowledge on the system performance. It was shown that, while noisy CSI leads to slight degradation in the BER performance, outdated CSI can cause detrimental performance degradation if the order of the users' channel gains change as a result of mobility.
\end{abstract}


\section{Introduction}
Driven by the vast increase in the global demand for increased wireless data connectivity and the recent advances  in solid-state lighting, visible light communication (VLC) has evolved significantly as a potential candidate to the wireless data explosion dilemma. Based on this, it has recently attracted the attention of both academia and industry   as an effective  complementing  technology to traditional radio frequency (RF) communications \cite{VLC_mag, haas_book, VLCmarket,haas3}. %VLC  is a form of optical wireless communication (OWC) that operates over the visible part of the electromagnetic spectrum, which is a completely untapped  license-free resource.
In VLC systems,   light emitting diodes  (LEDs) are utilized for  data transmission, where the intensity of the LED light is modulated at particularly  high switching rate that cannot be perceived by the human eye. This process is known as intensity modulation (IM). Then, at the receiver site,  a photo detector (PD) is used to convert the variations in the received light intensity into electrical current that is subsequently used for data recovery \cite{state-of-the-art,haas2,murat1}.

\subsection{Related Literature}
%The key limitation in  VLC systems is the narrow modulation bandwidth of the light sources. For instance, off-the-shelf white-light LEDs can provide limited switching rate of about 2-20 MHz \cite{LED1}.
%This restricts the realizable data rates of VLC systems. To solve the bandwidth limitation problem, multiple-input-multiple-output (MIMO) techniques  have been widely investigated for indoor VLC systems \cite{MIMO_VLC_Haas, joint, MIMO7, MIMO8}. MIMO transmissions are particularly  feasible in VLC systems since multiple LED arrays typically exist in indoor spaces. Moreover, MIMO configurations can reduce the difficulties in achieving physical alignment between the transmitting LED and the detector, in scenarios where the receiver moves around the coverage area \cite{mimo1}.
%

As a promising   broadband technology, VLC is expected to   provide remarkably high speed indoor communication and support ubiquitous connectivity. To this end, several   multiple access schemes have been  proposed for VLC systems, including, carrier sense multiple access (CSMA), orthogonal frequency division multiple access (OFDMA) and code division multiple access (CDMA) \cite{survey2}. In more details, in  CSMA-based systems, each LED is required to sense the channel before attempting to transmit in order to avoid collisions. Yet, when a LED's transmitted signal  is undetectable by others, the “hidden terminal” problem  arises leading to considerable performance degradation \cite{CSMA1}. A full-duplex carrier sense multiple access (FD-CSMA) protocol was proposed in \cite{CSMA_FD} to avoid  the “hidden terminal” problem. This is realised by considering   downlink transmissions  as  busy medium for the uplink channel, leading to reduced collisions.

However, carrier sensing needed for CSMA is not trivial in VLC due to line-of-sight (LOS) transmissions along with  the need for request-to-send/clear-to-send (RTS/CTS) symbols,  which increases the overall signalling overhead \cite{csma2}.

In the same context, OFDMA is an effective multiple access scheme but it  cannot be applied directly  to VLC systems due to the restriction of positive and real signals imposed by IM and the illumination requirements. Based on this, DC-biasing and clipping techniques have been proposed to adapt OFDMA to VLC systems. To this end, optical orthogonal frequency division multiple access (O-OFDMA) was proposed  as a modified  OFDMA scheme by  asymmetrically clipping the transmitted OFDM signal  at zero level in order to satisfy the positivity constraint  \cite{OFDMA1}. The performance of O-OFDMA was compared to optical orthogonal frequency division multiplexing interleave division multiple access (O-OFDM-IDMA); it was shown that while  O-OFDM-IDMA outperforms O-OFDMA in terms of power efficiency,  O-OFDMA exhibits the benefit of reduced decoding complexity. Moreover,  O-OFDMA provides lower peak-to-average power ratio (PAPR) compared to O-OFDM-IDMA. Yet,  adapting OFDMA in order to cope with  VLC requirements  leads to significant reduction of  spectral efficiency, which is a disadvantage in emerging communications that require significantly increased throughputs.

Likewise, CDMA is another multiple access scheme that has been proposed for VLC systems,  exploiting optical orthogonal codes (OOC) \cite{OOC1,OOC2} to allow multiple users to access the channel using full spectrum and time resources. Thus, it can provide enhanced spectral efficiency compared to  OFDMA. An experimental optical code-division multiple access (OCDMA) based VLC system was presented in \cite{CDMA1}.  It was shown that the  major drawback of OCDMA is the need for long OOC codes, which leads to a reduction of the theoretically  achievable data rates. This issue was addressed in \cite{cdma4} where a code-cycle modulation (CCM) technique was proposed  to enhance the spectral efficiency of OCDMA by using different cyclic shifts of the spreading sequence assigned to the different users. Yet another limitation of OCDMA is the poor correlation characteristics in the OOC codes.  This problem was addressed  in  \cite{CDMA2} by means of  a synchronization  mechanism that aims to  improve the performance of OCDMA.

Non-orthogonal multiple access (NOMA) was proposed in \cite{NOMA_VLC} as a spectrum-efficient multiple access scheme for downlink VLC systems. In NOMA, the signals of  different users  are superimposed in the power domain by allocating different power levels based on the channel conditions of each user. Thus, users can share the entire frequency and time  resources leading to increased  spectral efficiency. NOMA allocates higher  power levels  to users with worse channel conditions than those with good channel conditions. As a result, the user with the highest allocated power will be  able to decode directly its signal, while treating the signals of other users as noise. The other users in the system  perform  successive interference cancellation (SIC) for the multi-signal separation, prior to decoding their signals.
The concept of NOMA has recently gained interest in RF systems as a candidate for 5G and long term evolution–advanced (LTE-A) systems \cite{whiltepaper, NOMA-5g}.
The performance of NOMA downlink system with randomly deployed users was investigated in \cite{noma1}, where NOMA was shown to provide improved outage probability when power allocation is carefully designed. However, it was shown that the performance gains of NOMA are degraded in low signal-to-noise-ratio (SNR) scenarios. In \cite{NOMA-relay}, NOMA has been applied to  multiple-antenna relaying networks, where  the outage behavior of the mobile users has been analyzed. It was shown that NOMA leads to improved spectral efficiency and  fairness compared to orthogonal multiple access schemes.  The application of NOMA to  multiple-input-multiple-output (MIMO) configrations has been studied in \cite{NOMA-MIMO1, NOMA-MIMO2}. In MIMO-NOMA, two types of interference cancelation techniques are required: 1) SIC for the intra-beam demultiplexing of users having the same precoding weights, and 2) inter-beam interference cancellation for users with different precoding weights. The impact of channel state information (CSI)  on the performance of  NOMA was first examined in \cite{NOMA-SCI1}, where the ergodic capacity maximization problem was considered under  total transmit power constraints.  In \cite{NOMA-CSI2},  outage probability expressions for a downlink NOMA system have been derived under different channel uncertinity models based on imperfect CSI and second order statistics.

The performance of NOMA-VLC was analyzed in \cite{NOMA_VLC1,NOMA_haas} in terms of coverage probability and ergodic sum rate and it was  shown that NOMA enhances the system capacity compared to time division multiple access (TDMA). Likewise, it was shown in  \cite{NOMA_VLC2} that  NOMA also  outperforms  OFDMA in terms of the achievable data rate in the downlink of VLC systems.
It is noted here that the majority of the reported contributions  on optical NOMA  assume perfect knowledge of the channel fading coefficients. However, in  practical communication scenarios, these coefficients must be first estimated and then used in the detection process. Yet,  channel estimation  can not  always be perfect in practice, which results to subsequent decoding  errors that degrade the overall system performance. VLC channel estimation errors have been considered in \cite{Joint} for the  joint optimization of precoder and equalizer in  optical MIMO systems, while the impact of noisy CSI on the performance of different MIMO precoding schemes was investigated in \cite{marshoud}. Likewise, the contribution  in \cite{coordinated} analyzed the impact of CSI errors on  a multiuser VLC downlink network, where the corresponding system performance was evaluated under two channel  uncertainty models, namely  noisy and outdated CSI.  %In this paper, we  analyze the impact of channel estimation errors on the bit error performance of  NOMA-VLC systems, where we consider noisy and outdated channel uncertainty models for the indoor VLC system.  To the best of our knowledge, the impact of channel estimation error   has not been studied before in the context of a VLC-NOMA  system.

It is also recalled that in order to consider  NOMA in commercial implementations, it should be first ensured that it satisfies  the requirements imposed by the illumination functionality of VLC systems \cite{dimming4}.  According to to the IEEE Standard $802.15.7$ \cite{ook_standard}, VLC wireless networks are required to support light dimming, allowing the control of the perceived light brightness  according to the users' preference. Hence, integrating  efficient dimming techniques into VLC systems is vital for energy savings as well as for aesthetic and comfort purposes, rendering a wide implementation of VLC systems more rational. Various dimming methods have been proposed in the literature to incorporate  data transmission  into dimmable  light intensities by means of different modulation and coding schemes.
In general,  brightness control can be achieved by two different techniques: continuous current reduction (CCR) and pulse-width modulation (PWM). In CCR, also known as analog dimming,  dimming control is realized
by changing the forward current of the LED,  while in  PWM, also known as digital dimming,  the forward current remains constant while the duty cycle of the signal is varied in order to meet the dimming requirements. Analog and digital dimming have been applied to different modulation schemes in VLC, such as on-off keying (OOK) and pulse position modulation (PPM) \cite{dimming3}. Moreover,  dimming control was implemented in a VLC-OFDM system in \cite{dimming5} where it was shown  that CCR achieves higher luminous efficacy, whereas  PWM  leads to throughput degradation.  It is noted here that supporting  dimming control in NOMA-VLC systems can be rather challenging because NOMA is fundamentally based on dividing  the LED power among the different users in the network, which renders the corresponding  system performance highly sensitive to any reduction in transmit power. Based on the above, it is shown  in detail that the present work quantifies  the impact of channel estimation errors on the  performance of indoor NOMA-VLC systems under noisy and outdated channel uncertainty models as well as under analog and digital dimming techniques. A detailed outline of the contribution of this work is provided in the following subsection.

\subsection{Contribution}
In the present paper, we consider NOMA as a multiple access scheme for the downlink of indoor VLC  networks  due to the following advantageous characteristic \cite{NOMA_VLC}:
\begin{itemize}
\item NOMA is efficient in  multiplexing a small number of users, which is the case in VLC sytems where an LED is regarded as a small cell that serves few users in room environments.
\item NOMA provides superior performance gain at high SNR scenarios \cite{noma1}, and thus, it is suitable for VLC links as they  typically operate at rather high SNR due to the existence of strong LOS component and a short propagation distance.
\end{itemize}

In this context, the contribution of this paper is summarized below:
\begin{enumerate}
  \item We investigate the error rate  performance of NOMA-VLC systems, and derive an exact analytic expression for the  bit-error-rate (BER) for an arbitrary number of users  for the case of perfect CSI.
  \item We quantify the impact of CSI errors on the system performance under two channel uncertainty models, namely   noisy CSI, and outdated CSI that may result from the mobility of the indoor users. In this context, we  derive a closed-form approximated expression   for the BER under noisy CSI, as well as a tight  upper bound for the BER under outdated CSI.
  \item We analyze the effect of dimming support on the BER performance of NOMA-VLC systems. In particular, two different dimming schemes are considered, OOK analog dimming and variable OOK (VOOK) dimming.

\end{enumerate}

To the best of our knowledge, the above topics have not been previously  investigated in the open technical literature, including similar analyses in conventional radio communications.

\subsection{Structure}
The remainder of the paper is organized as follows: Section \ref{sec:model} describes the channel and system model of an indoor VLC downlink network. Section \ref{sec:analysis} analyzes the BER performance under perfect CSI, while Section \ref{sec:analysis_imperfect} investigates the performance of NOMA under two different cases for CSI errors. Numerical results and related discussion are presented  in Section \ref{sec:results} while closing remarks are provided   in  Section \ref{sec:conc}.

\section{Channel and System Model}
\label{sec:model}
We consider a single-LED downlink VLC system deployed in an indoor environment. The LED has a dual function of illumination and communication, and serves $N$ users  simultaneously, by modulating the intensity of the emitted light according to the data received through a power line communications (PLC) backbone network. Also,  all  users are equipped with a single PD, that performs direct detection to extract the transmitted signal from the received optical carrier. This is realized considering unipolar OOK modulation  due to its popularity  in VLC systems \cite{ook_standard,murat2}.

\subsection{The VLC Channel}
The current set up is based on LOS communication scenarios  as illustrated in Fig. \ref{fig:channel},  since multipath delays resulting from reflections and diffuse refractions are typically negligible in indoor VLC settings \cite{modeling_LED}.  The channel between user $U_i$ and the corresponding LED is given by
\begin{equation}\label{equ:hi}
h_{i}= \begin{cases}\frac{A_i}{d^2_{i}}R_o(\varphi_{i})T_s(\phi_{i})g(\phi_{i})\cos(\phi_{i}), \qquad & 0\leq\phi_{i}\leq \phi_{c}\\
0,&\phi_{i}>\phi_{c}
\end{cases}
\end{equation}
\noindent where $i$ =1, 2, 3, \dots ,$N$, $A_i$ represents the receiver PD area, $d_{i}$ accounts for the distance between the transmitting LED  and the \emph{i}-th receiving PD, $\varphi_{i}$ is the angle of emergence with respect to the transmitter axis, $\phi_{i}$  is the angle of incidence with respect to the receiver axis, $\phi_{c}$ is the field of view (FOV) of the PD, $T_s(\phi_{i})$ is the gain of optical filter and $g(\phi_{i})$ is the gain of the optical concentrator, which is expressed as

\begin{figure}[h]
\center
\includegraphics[width=3.5in,height=2.5in]{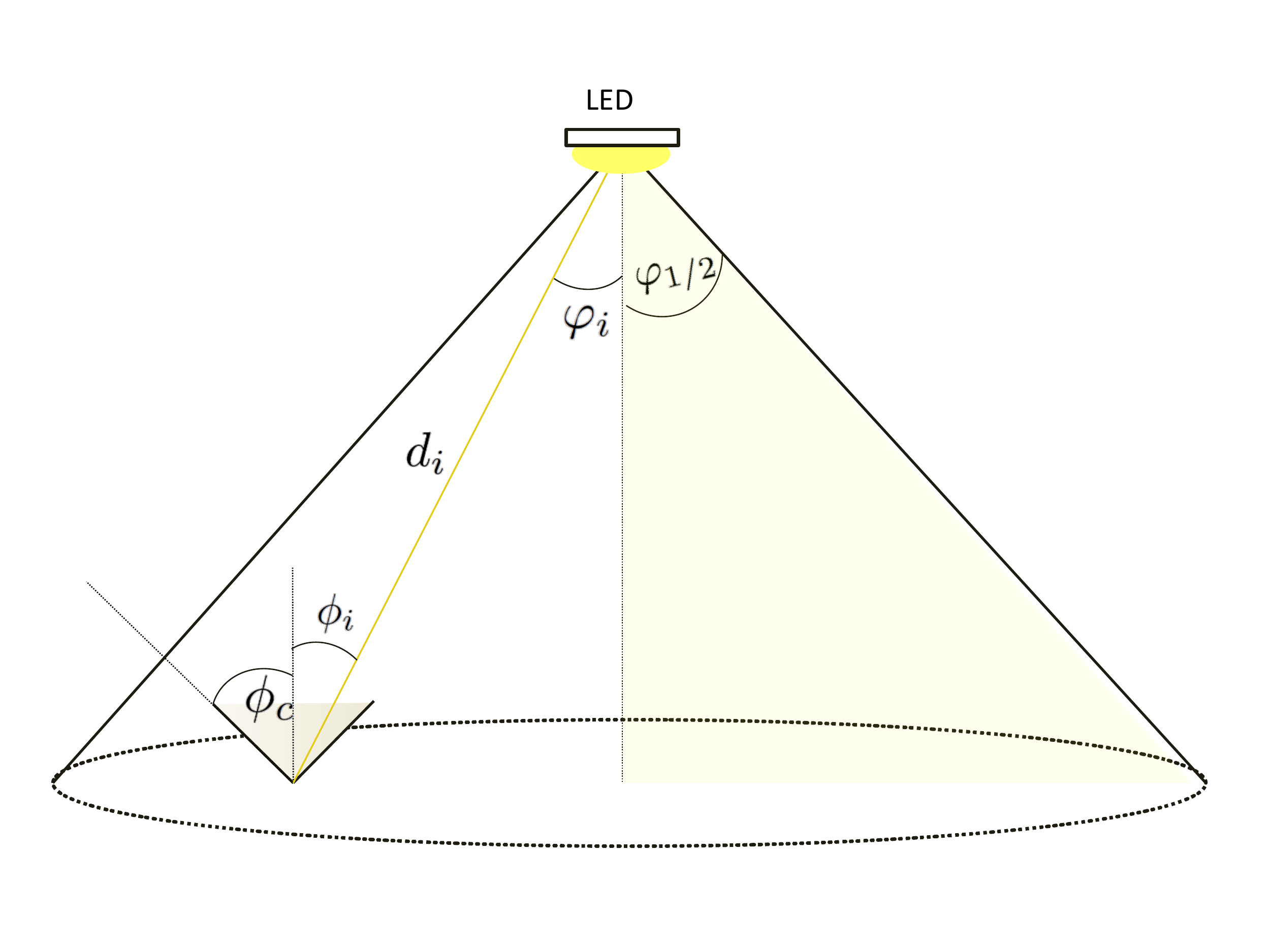}
\caption{VLC channel model.}
\label{fig:channel}
\end{figure}

\begin{equation}\label{equ:g}
g(\phi_{i}) = \begin{cases} \frac{n^2}{\sin^2(\phi_{c})},\qquad \qquad & 0\leq\phi_{i}\leq \phi_{c}\\
 0, \qquad &\phi_{i}>\phi_{c}
 \end{cases}
\end{equation}
\noindent where $n$ denotes the corresponding  refractive index. Moreover, $R_o(\varphi_{i})$ in (\ref{equ:hi}) is the Lambertian radiant intensity of the transmitting LEDs, which  can be expressed as
\begin{equation}\label{equ:R_o}
R_o(\varphi_{i}) = \frac{(m+1)}{2\pi}\cos^m(\varphi_{i})
\end{equation}
where $m$ is the order of Lambertian emission, calculated as
\begin{equation}\label{equ:m}
 m =\frac{- \ln(2)}{\ln(\cos(\varphi_{1/2}))}
\end{equation}
with $\varphi_{1/2}$ denoting  the transmitter semi-angle at half power.
To this effect, the  receiver-site noise is drawn from   a circularly-symmetric Gaussian distribution of zero mean and variance
\begin{equation}\label{equ:totalnoise}
\sigma_{n}^2 = \sigma_{sh}^2 + \sigma_{th}^2
\end{equation}
where $\sigma_{sh}^2 $ and $ \sigma_{th}^2$ are the variances of the shot noise and  thermal noise, respectively.

The shot noise in an optical wireless channel results from the high rate physical photo-electronic conversion process,  with variance at the \emph{i}-th PD
\begin{equation}\label{equ:shot}
\sigma_{sh_i}^2=2qB\left(\gamma h_{i}x_{j}+I_{bg}I_2 \right)
\end{equation}
\noindent where \textit{q} is the electronic charge, $\gamma$ is the detector responsivity, \textit{B} is the corresponding bandwidth, $I_{bg}$ is background current, and $I_2$ is the noise bandwidth factor. Furthermore,  the thermal noise is generated within the transimpedance receiver circuitry and its variance is given by
\begin{equation}\label{equ:thermal}
\sigma_{th_i}^2 = \dfrac{8 \pi K T_k}{G}\eta AI_2B^2 + \dfrac{16 \pi^2K T_k \Gamma }{g_m}\eta^2 A^2 I_3 B^3
\end{equation}
where $K$ is Boltzmann's constant, $T_k$ is the absolute temperature, \textit{G} is the open-loop voltage gain, \textit{A} is the PD area, $\eta$ is the f\mbox{}ixed capacitance of the PD per unit area, $\Gamma$ is the field-effect transistor (FET) channel noise factor, $g_m$ is the FET transconductance, and $I_3$ = 0.0868 \cite{Fundamental}.

\begin{figure*}[   \widetilde{}]
\center
\vspace{-2cm}
\includegraphics[width=6in,height=5in]{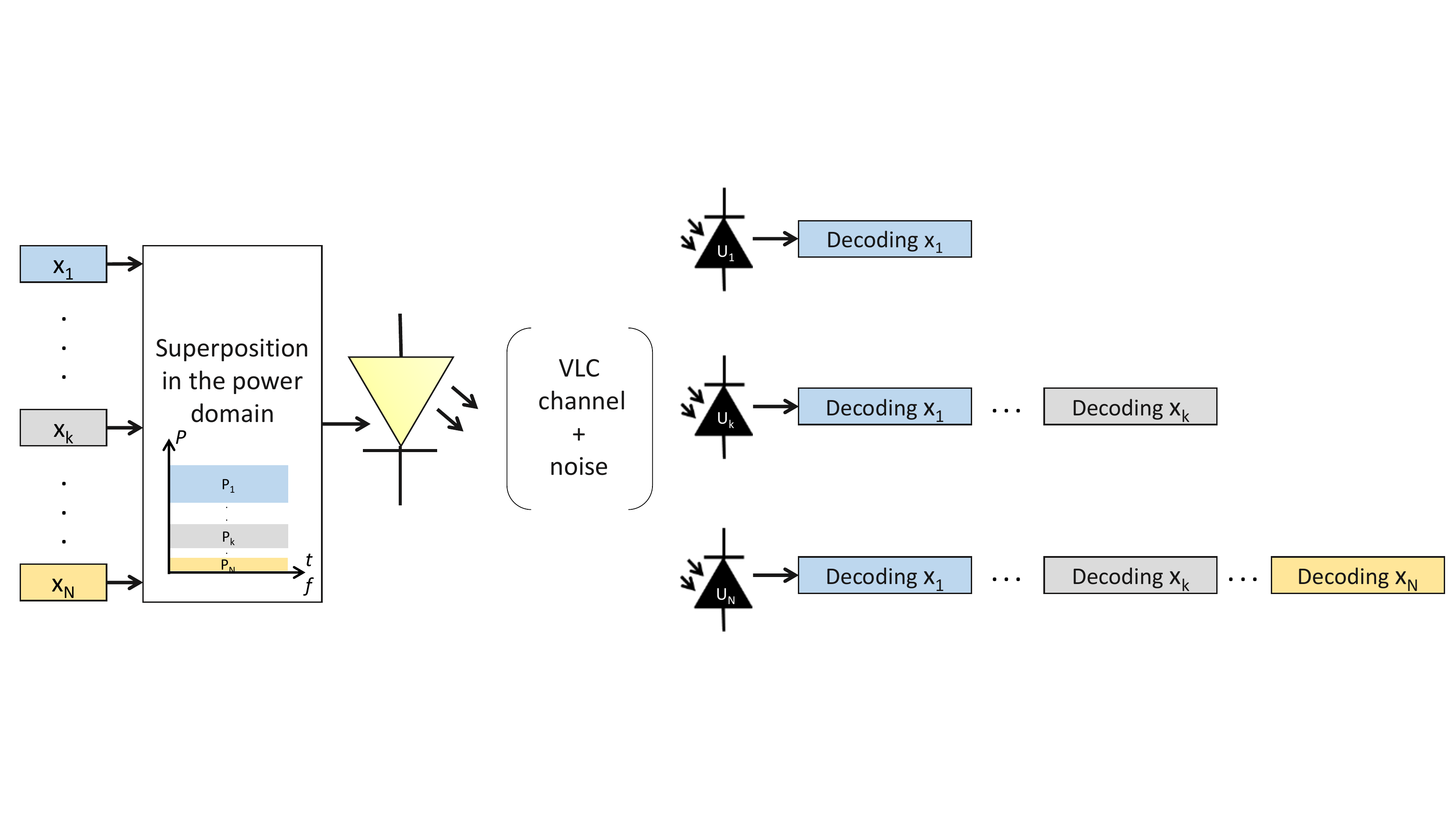}
\vspace{-3cm}
\caption{NOMA-VLC downlink. }
\label{fig:system_model}
\end{figure*}

\subsection{NOMA Transmission}
Without loss of generality, we assume that the  users $U_1$, \ldots , $U_{N}$ are sorted in an ascending order according to their channels, i.e.   $h_1 \leq h_2 \leq \cdots \leq h_N$.  Using NOMA, the LED transmits the real and non-negative signals
$s_1$, \ldots , $s_N$  with associated power values $P_1$, \ldots , $P_N$, where $s_i$ conveys information intended for user $U_i$, as shown in Fig. \ref{fig:system_model}. Unless otherwise stated, the term power  refers to the optical power which is directly proportional to the LED driving current.  To this effect, the $N$ transmitted signals are superimposed in the power domain as follows:
\begin{equation}\label{equ:transmitted}
x=\sum \limits_{i=1}^{N} P_i s_i
\end{equation}
and the  LED total transmit power is $P_{LED}= \sum_{i = 1}^{N}  P_i$. At the PD site, direct detection of the received signal is performed based on the received optical power and the  received signal at user $U_k$ can be expressed as
\begin{equation}\label{equ:received}
y_k= \gamma h_k \sum \limits_{i=1}^{N} P_i s_i     + n_k
\end{equation}
where  $\gamma$ is the detector responsivity and $n_k$ denotes zero-mean additive white Gaussian noise (AWGN) with variance $\sigma_n^2$. Henceforth,   $\mathcal{N}(\mu,\sigma^2)$ represents the probability density function (PDF) of Gaussian distribution with mean $\mu$ and variance $\sigma^2$.
%\begin{equation}
%\mathcal{N}_x(\mu,\sigma^2)= \frac{1}{\sqrt{2 \pi \sigma^2}} e^{-\frac{{x- \mu}^2}{2 \sigma^2}}.
%\end{equation}
Based on this and assuming   unipolar OOK signals, the PDF of the received signal at $U_k$ can be represented as
%\begin{equation}
%\begin{cases}
%f_{Y_k|S_k} (y_k | s_k =0) = \frac{1}{\sqrt{2 \pi \sigma_n^2}} e^{-\frac{\left(y_k - h_k \left( \sum \limits_{i=1,i\neq k }^{N}   P_i s_i \right) \right)^2}{\sigma_n^2}}\\
%f_{Y_k|S_k} (y_k | s_k =1) = \frac{1}{\sqrt{2 \pi \sigma_n^2}} e^{-\frac{\left(y_k - h_k \left(P_m + \sum \limits_{i=1,i\neq k }^{N}   P_i s_i \right) \right)^2}{\sigma_n^2}}.
%\end{cases}
%\end{equation}

\begin{equation}
\begin{cases}
f_{Y_k|S_k} (y_k | s_k =0) = \mathcal{N}_{y_k}( \gamma h_k \sum \limits_{\substack{i=1\\i\neq k} }^{N}   P_i s_i ,{\sigma_n}^2)\\
f_{Y_k|S_k} (y_k | s_k =1) = \mathcal{N}_{y_k}( \gamma h_k (P_k + \sum \limits_{\substack{i=1\\i\neq k} }^{N}   P_i s_i ) ,{\sigma_n}^2).
\end{cases}
\end{equation}

It is recalled that the  multi-user interference at user $U_k$ can be eliminated by means of SIC. Based on this, in order to decode its own signal, $U_k$ needs to successfully  decode and subtract the signals of all other users with lower decoding order, i.e. $s_1$, \ldots , $s_{k-1}$. As a result, the  residual interference from $s_{k+1}, \ldots , s_{N}$ becomes insignificant and can be treated as noise.  %The achievable data rate for the \emph{k}-th user is given by
%\begin{equation}
%R_k = B \log_2\left(1+\frac{ (\gamma h_k P_k)^2}{\sum_{j=1}^{k-1} (e_j \gamma h_k P_j)^2 + \sum_{l=k+1}^{N} (\gamma h_k P_l)^2 + {\sigma_n}^2 }\right),
%\end{equation}
%where $e_j$ represents the potential residual interference caused by detection error in the decoding of the signal $s_j, 1\leq j \leq k-1$.

In order to facilitate SIC decoding, the LED allocates higher  transmission power to users with poor channel gains. The simplest power allocation scheme is the fixed power allocation (FPA), where  the associated  power of the \emph{i}-th sorted user is set to
\begin{equation}
\label{equ:FPA}
P_i = \rho P_{i-1}
\end{equation} with $\rho$ denoting  the power allocation factor ($0 < \rho < 1$).
According to FPA, the  power allocated to user $U_i$ is reduced at the increase of $h_i$  because  users with good channel conditions require lower power levels  to successfully decode their desired signals, after canceling the interference from  the signals of the users with lower decoding order. This is the fundamental principle of NOMA and has been shown to provide remarkable performance gains in RF-based communications \cite{NOMA-5g}.  %%In \cite{NOMA_VLC}, we have proposed a novel power allocation strategy which we called gain ratio power allocation (GRPA),    that considers not only the order, but also the relative values of users' channel gains in order to ensure fair power allocation. Once the decoding order is set, the  power allocated to  the $i$th sorted user is
%\begin{equation}\label{equ:GRPA}
%P_{i}=(\frac{h_{1}}{h_i})^i P_{i-1}.
%\end{equation}
\section{NOMA-VLC with  perfect CSI}
\label{sec:analysis}
It is recalled that accurate CSI is of paramount importance in conventional and emerging communications as encountered imperfections in practical deployments  lead to significant degradation of the overall system performance. This is also the case  in VLC; therefore, in  this section, we derive a closed form expression for the  BER   of  NOMA-VLC systems employing unipolar OOK under the assumption of perfect knowledge of the channel coefficients  and ideal time synchronization.

\begin{theorem}
Given that user $U_k$  attempts to cancel  the first $k-1$ signals from the aggregate received signal in succession, the BER of  $U_k$  achieved by the NOMA scheme can be written as
\begin{equation}
\begin{aligned}\label{equ:BER1}
\mbox{Pr}_{e_k} = \sum_{e_{k-1= -1,0,1}}  \ldots     \sum_{e_{1= -1,0,1}}    \mathcal{P}(e_{k-1}|e_1, \dots ,e_{k-2})
 \enspace \mathcal{P}(e_{k-2}|e_1, \ldots ,e_{k-3})  \enspace  \cdots \enspace  \mathcal{P}(e_1)  \mbox{Pr}_{e_k}|e_1, \ldots ,e_{k-1}
\end{aligned}
\end{equation}

It can be inferred from (\ref{equ:BER1})  that    the probability of error in decoding the $k^{th}$ user signal depends on the detections of the signals $1$ to $k-1$ that are performed during SIC stages. So the probability of error in decoding the $k^{th}$ signal equals the  conditional error probability  in decoding signal $k$, conditioned on the error probabilities of the  previous detection stages $(\mbox{Pr}_{e_k}|e_1, \ldots ,e_{k-1})$, multiplied by the error probabilities of all the previous  detection stages,
where $ \mathcal{P}({e_i}|{e_1}, \ldots ,{e_{i-1}})$ is the bit-error probability (BEP) of the $i^{th}$  detection stage conditioned on the  previous $1$ to $i-1$ BEPs, which is represented as

\begin{equation*}
\label{equ:BER1step2}
\mathcal{P}({e_i}|{e_1}, \ldots ,{e_{i-1}}) = \begin{cases} 1-{\rm Pr}_{{e_i}|{e_1}, \ldots ,{e_{i-1}}} & {e_i} = 0 \\ \frac{1}{2}{\rm Pr}_{{e_i}|{e_1}, \ldots ,{e_{i-1}}} & {e_i} = -1,1
 \end{cases}
\end{equation*}
with ${e_i} = \hat{s_j} - s_j $ denoting the error in detecting the \emph{i}-th OOK signal, %which can be written as
%\begin{equation*}
%{e_i}_\iota = \begin{cases} -1 & \iota=1 \\ 0 & \iota=2 \\ 1 & \iota=3 \end{cases},
%\end{equation*}
and $\mbox{Pr}_{{e_k}}|{e_1}, \ldots ,{e_{k-1}}$ is the error probability in decoding the \emph{k}-th signal conditioned on the previous detections, namely
\begin{equation}\label{equ:BER2}
\begin{split}
\mbox{Pr}_{{e_k}}|{e_1}, \ldots ,{e_{k-1}}   =& \frac{1}{2^{N-k+1}} \sum_{i=1}^{2^{N-k}} \mathcal{Q}\left(\frac{  \gamma h_k}{\sigma_n} \left(\frac{P_k}{2} - \sum_{j=1}^{k-1} {e_j} P_j - \sum_{l=k+1}^{N} P_l A_{il}\right)\right) \\
& +  \frac{1}{2^{N-k+1}} \sum_{i=1}^{2^{N-k}} \mathcal{Q}\left(\frac{ \gamma h_k}{\sigma_n} \left(\frac{P_k}{2} + \sum_{j=1}^{k-1} {e_j} P_j + \sum_{l=k+1}^{N} P_l A_{il}\right) \right)
\end{split}
\end{equation}
where $\mathcal{Q}(x)=\frac{1}{\sqrt{2 \pi}} \int_{x}^{\infty} e^{-\frac{y^2}{2}} dy$ denotes  the one dimensional Gaussian $Q-$function and the term $\sum_{j=1}^{k-1} {e_j} P_j$ represents the potential residual interference caused by detection error in the decoding of  $s_1$, \ldots , $s_{k-1}$. Moreover, $\sum_{l=k+1}^{N} P_l A_{il}$ corresponds  to the interference  caused by  $s_{k+1}$, \ldots , $s_N$, where the elements of the matrix

\begin{equation}
\label{equ:A}
\begin{aligned}
\textbf{A}= \left[\begin{array}{ccc}
    A_{1\enspace k+1}       &  \dots & A_{1\enspace N}   \\
    A_{2\enspace k+1}      &  \dots & A_{2\enspace N}   \\
        \vdots        &  \vdots &\vdots    \\
    A_{2^{N-k}\enspace k+1}       &  \dots & A_{2^{N-k}\enspace N} \\
\end{array} \right]
= \left[\begin{array}{cccc}
    0     & 0 & \dots  & 0  \\
    0         & 0 &  \dots  & 1   \\
         \vdots    & \vdots    &\vdots  &\vdots  \\
    1       & 1 &   \dots  & 1 \\
\end{array} \right].
\end{aligned}
\end{equation}
demonstrate the possible combinations of interference depending on the transmitted OOK vectors.
\end{theorem}
\begin{proof} The proof is provided in  Appendix A. \end{proof}

\section{NOMA-VLC with Imperfect CSI}
\label{sec:analysis_imperfect}
As already mentioned, NOMA  configurations are rather sensitive to  the  knowledge of all users' channel coefficients. This is of paramount importance not only  for successful data recovery at the receivers, but it is also  crucial at the transmitter site for determining the  power to be allocated to each corresponding user. This is based on  the fact that users must receive signals with different  power levels, depending on the ordering of their channel gains, in order to effectively facilitate SIC.     Thus, while the channel in VLC is technically deterministic for specific transmitter-receiver specifications and fixed locations, the assumption of perfect CSI is not practically realistic even for indoor VLC systems. Typically, CSI  can be firstly determined at the receiver site with the aid of  periodic pilot signals and then, the receivers feed back the quantized channel coefficients to the transmitters through an RF or infrared (IR) uplink\footnote{Although VLC uplink can be theoretically possible, it is energy-inefficient for low-power mobile devices. Thus, utilizing uplink-downlink reciprocity for acquiring  CSI at the transmitter is not practically relevant in the context of VLC systems. }.  To this effect, the  uncertainty in the VLC channel estimation arises from the noise in the downlink and uplink channels as well as from the mobility of users in indoor environments. Moreover,  AD/DA conversion of the channel estimates introduces quantization errors that add to the channel uncertainty, which is beyond the scope of our work \cite{quantization}.
Therefore, it becomes evident that it is essential to quantify the effects of imperfect CSI on the performance of NOMA VLC systems. To this end,  we consider two different  realistic stochastic uncertainty models for the CSI, namely noisy CSI and outdated CSI.
\subsection{Noisy CSI}
By assuming that  $\hat{h}_k$ is the estimate  for the channel between the \emph{k}-th user and the transmitting LED, it follows that
\begin{equation}\label{equ:noisyCSI}
\hat{h}_k= h_k + \epsilon_{n}
\end{equation}
where $\epsilon_n$ denotes the  channel estimation error modeled as  a  zero-mean Gaussian distribution  with variance $\sigma_{\epsilon_{n}}^2$, i.e, $\epsilon_n \sim \mathcal{N}(0,\sigma_{\epsilon_{n}}^2)$,  which has been
adopted  as a reasonable model for indoor VLC systems \cite{Joint, coordinated}. To this effect, it immediately follows that  the channel estimate  $\hat{h}_k$ can be modelled as $\hat{h}_k \sim \mathcal{N}(h_k,\sigma_{\epsilon_{n}}^2)$.  Under the realistic case of  noisy CSI, the conditional error probability for user $U_k$ can be approximated by the following Proposition.

\begin{proposition}
Under noisy CSI, the error in decoding the \emph{k}-th signal at $U_k$ conditioned on the previous detections is given by
%\begin{equation}\label{equ:BER_noisy2}
%\begin{aligned}
%\mbox{Pr}_{{e_k}_\iota}|{e_1}_\iota, \ldots ,{e_{k-1}}_\iota   \approx \frac{1}{2^{N-k+1} \sqrt{2 \pi \sigma_{\epsilon_{n}}^2}}  \sum_{i=1}^{2^{N-k}}  \sqrt{\frac{\pi}{- \alpha_i}} e^{-\frac{\beta_i^2}{4 \alpha_i} + \zeta_i}
%+ \frac{1}{2^{N-k+1} \sqrt{2 \pi \sigma_{\epsilon_{n}}^2}}   \sum_{i=1}^{2^{N-k}}  \sqrt{\frac{\pi}{- \tilde{\alpha_i}}} e^{-\frac{\tilde{\beta_i}^2}{4 \tilde{\alpha}} + \tilde{\zeta_i}}.
%\end{aligned}
%\end{equation}
\begin{equation}\label{equ:BER_noisy2}
\begin{aligned}
\mbox{Pr}_{{e_k}}|{e_1}, \ldots ,{e_{k-1}}   & \approx \frac{1}{2^{N-k+1} \sqrt{ \sigma_{\epsilon_{n}}^2}}  \sum_{i=1}^{2^{N-k}}  {\frac{1}{\sqrt{ \alpha_i}}} e^{c-\frac{\mathcal{P} \left(2 a {h_k}^2 \mathcal{P}+b \left(b \mathcal{P} {\sigma_{\epsilon_{n}}}^2+2 h_k {\sigma_n}\right)\right)}{4 a \mathcal{P}^2{\sigma_{\epsilon_{n}}}^2-2{\sigma_n}^2}}\\
&+ \frac{1}{2^{N-k+1} \sqrt{ \sigma_{\epsilon_{n}}^2}}  \sum_{i=1}^{2^{N-k}}  {\frac{1}{\sqrt{ \alpha_i}}} e^{c-\frac{\tilde{\mathcal{P}} \left(2 a {h_k}^2 \tilde{\mathcal{P}}+b
   \left(b \tilde{\mathcal{P}} {\sigma_{\epsilon_{n}}}^2+2 h_k {\sigma_n}\right)\right)}{4 a \tilde{\mathcal{P}}^2{\sigma_{\epsilon_{n}}}^2-2{\sigma_n}^2}}
\end{aligned}
\end{equation}
where
\begin{equation}
\alpha =\frac{1}{{\sigma_{\epsilon_{n}}}^2}-\frac{2 a\mathcal{P}^2}{{\sigma_n}^2}
\end{equation}
and
\begin{equation}
\mathcal{P}=\frac{P_k}{2} - \sum_{j=1}^{k-1} {e_j} P_j - \sum_{l=k+1}^{N} P_l A_{il}
\end{equation}
while, similarly
\begin{equation}
\tilde{\alpha} =\frac{1}{{\sigma_{\epsilon_{n}}}^2}-\frac{2 a\tilde{\mathcal{P}}^2}{{\sigma_n}^2}
\end{equation}
and
\begin{equation}
\tilde{\mathcal{P}}= \frac{P_k}{2} + \sum_{j=1}^{k-1} {e_j} P_j + \sum_{l=k+1}^{N} P_l A_{il}.
\end{equation}
\end{proposition}

\begin{proof} The proof is provided in Appendix B.
\end{proof}

%\begin{figure*}[!t]
%\begin{equation}\label{equ:BER_noisy3}
%\begin{aligned}
%& \mbox{Pr}_{e_k}|e_1, \ldots ,e_{k-1}   \approx\\ & \frac{1}{2^{N-k+1} \sqrt{2 \pi \sigma_{\epsilon_{n}}^2}} \int_{- \infty}^{\infty} \sum_{i=1}^{2^{N-k}}   e^{- \frac{\epsilon_{n}^2}{2 \sigma_{\epsilon_{n}}^2}} e^{a \bigg(\frac{h_k + \epsilon_{n}}{\sigma_n} (\frac{P_k}{2} - \sum_{j=1}^{k-1} e_j P_j - \sum_{l=k+1}^{N} P_l A_{il})\bigg)^2 + b \bigg(\frac{h_k + \epsilon_{n}}{\sigma_n} (\frac{P_k}{2} - \sum_{j=1}^{k-1} e_j P_j - \sum_{l=k+1}^{N} P_l A_{il})\bigg) + c}
% \\   +  &\frac{1}{2^{N-k+1} \sqrt{2 \pi \sigma_{\epsilon_{n}}^2}} \int_{- \infty}^{\infty}  \sum_{i=1}^{2^{N-k}} e^{- \frac{\epsilon_{n}^2}{2 \sigma_{\epsilon_{n}}^2}} e^{a \bigg(\frac{h_k + \epsilon_{n}}{\sigma_n} (\frac{P_k}{2} + \sum_{j=1}^{k-1} e_j P_j + \sum_{l=k+1}^{N} P_l A_{il}) \bigg)^2 + b \bigg(\frac{h_k + \epsilon_{n}}{\sigma_n} (\frac{P_k}{2} + \sum_{j=1}^{k-1} e_j P_j + \sum_{l=k+1}^{N} P_l A_{il}) \bigg) + c}   .
% \end{aligned}
%\end{equation}
%\noindent\makebox[\linewidth]{\rule{520pt}{.4pt}}
%\end{figure*}
%

\subsection{Outdated CSI}
Outdated CSI error may result  from the variations in channel realizations   due to the mobility of users and/or shadowing effects that occur after the latest channel estimate update. In this context, we consider  a deterministically bounded random variable $\epsilon_o$ to model the outdated CSI error as
\begin{equation}\label{equ:outdatedCSI}
\hat{h}_k= h_k + \epsilon_{o}
\end{equation}
where $\epsilon_{o} \leq \mathcal{E}$, with $\mathcal{E}$ denoting the error bound that occurs when the mobile user moves with maximum velocity between the  reception of pilot signals and data \cite{coordinated}. In what follows, we  derive  a tight  upper bound for the conditional error probability at user $U_k$.

\begin{proposition}
The conditional error probability at user $U_k$ for the case of outdated CSI can be upper bounded as follows:
\begin{equation}\label{equ:BER3}
\begin{aligned}
 \mbox{Pr}_{e_k}|{e_1}, \ldots ,{e_{k-1}}   \leq  & \frac{1}{2^{N-k+1}} \sum_{i=1}^{2^{N-k}} \mathcal{Q}\left(-\frac{P_k}{2 {\sigma_n} }  \mathcal{E}+\frac{ \gamma h_k}{\sigma_n} \left(\frac{P_k}{2} - \sum_{j=1}^{k-1} {e_j} P_j - \sum_{l=k+1}^{N} P_l A_{il}\right) \right)
 \\ &  +  \frac{1}{2^{N-k+1}} \sum_{i=1}^{2^{N-k}} \mathcal{Q}\left(-\frac{P_k}{2{\sigma_n} }  \mathcal{E} +\frac{ \gamma h_k}{\sigma_n} \left(\frac{P_k}{2} + \sum_{j=1}^{k-1} {e_j} P_j + \sum_{l=k+1}^{N} P_l A_{il}\right) \right).
 \end{aligned}
\end{equation}
\end{proposition}
\begin{proof} The proof is provided in Appendix C.
\end{proof}
\begin{figure}[h]
\center
\includegraphics[width=3.5in,height=2.5in]{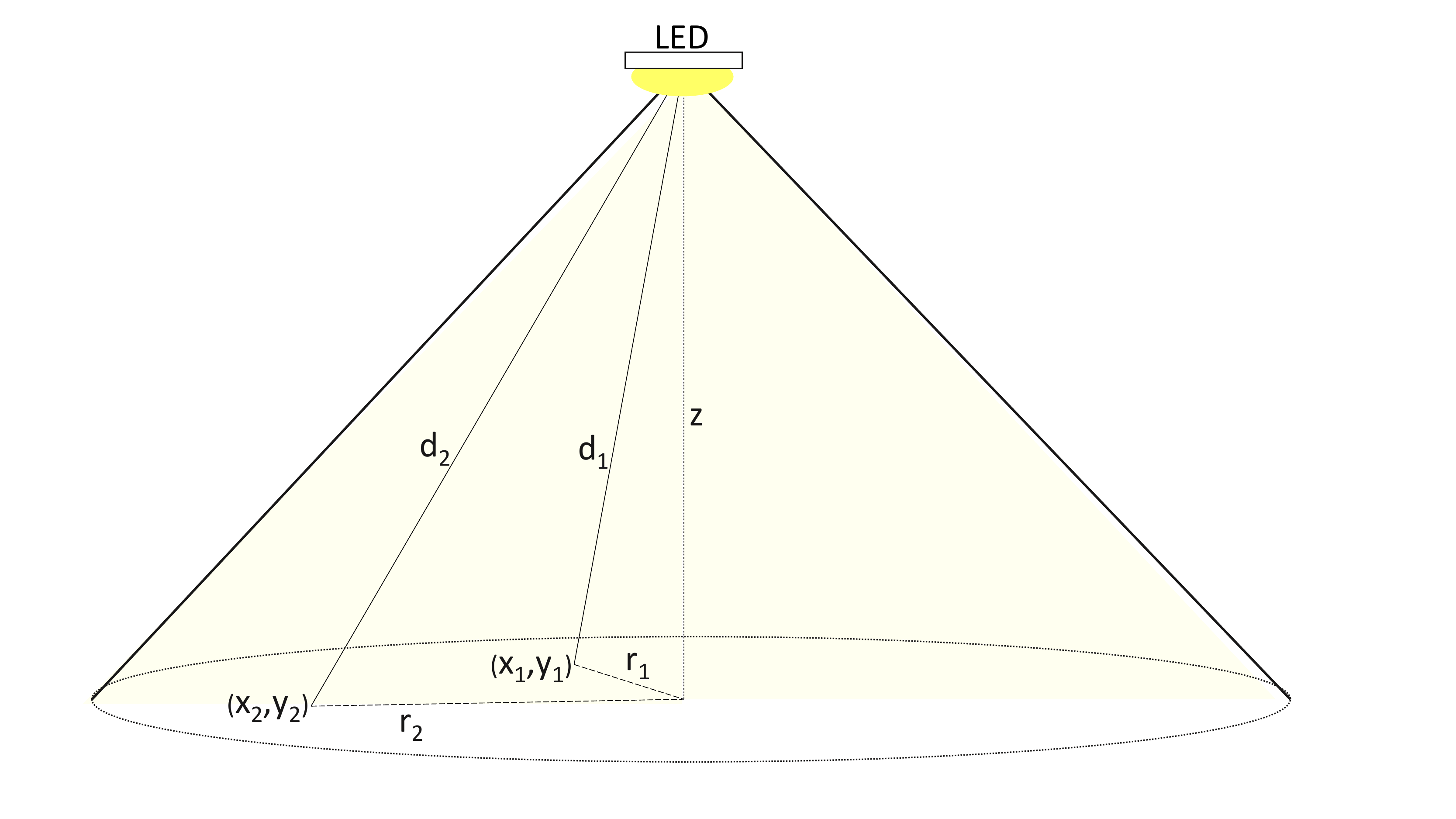}
\caption{Outdated CSI resulting from user mobility.}
\label{fig:mobility}
\end{figure}
\subsubsection{Determination of the value of $\mathcal{E}$}
In order to obtain the upper bound on the CSI error, we simplify the channel gain  by substituting \eqref{equ:g} and \eqref{equ:R_o} into \eqref{equ:hi} and substituting  $\cos\varphi_{i}$ with $z/d_{i}$, where $z$ denotes the  height between the LEDs and the PDs, which is assumed to be fixed, i.e., at a level of an ordinary table. Based on this and by  assuming  vertical alignment of LEDs and PDs, the corresponding  channel gain  $h_{i}$ can be expressed as
\begin{equation}\label{equ:channel_simple2}
h_{i} = \varpi \frac{1}{d^{m+3}_{i}}
\end{equation}
where
\begin{equation}
\varpi=\frac{(m+1) A_i  T_s(\phi_{i})g(\phi_{i})}{2\pi}.
\end{equation}

By now referring  to Fig. \ref{fig:mobility}, let user $U_k$ move in the horizontal plane from $(x_1,y_1)$ to $(x_2,y_2)$ with maximum velocity $v$. Then, the error bound $\mathcal{E}$ can be calculated as follows
\begin{equation}\label{equ:outdatedCSI3}
\mathcal{E} = \varpi |{d^{m+3}_{2}} - {d^{m+3}_{1}}|
\end{equation}
where $d^2_1 = r^2_1 + z^2$, $d^2_2 = r^2_2 + z^2$, and $v = \sqrt{r^2_1 - r^2_2}/t$, with $t$ denoting the time elapsed since the last CSI update. It is evident  that the algebraic representation of (\ref{equ:outdatedCSI3}) is tractable and can be computed straightforwardly. Furthermore, it is particularly accurate, as shown in detail in Section \ref{sec:results}.

\section{NOMA-VLC with Dimming Control}
In this Section, we investigate  the performance of NOMA-VLC under dimming control. To this end, it is firstly recalled that adjusting  the brightness level of the LED can be achieved by two approaches: 1) analog dimming, where the driving current of the LED is directly adjusted to the required illumination level;  2) digital dimming, in which the driving current is maintained constant while the duty cycle is varied in order to acquire the desired brightness \cite{dimming1, dimming2}. Therefore, analog dimming is straightforward  to implement given  that the LED brightness is directly proportional to the forward current. However, this technique may cause chromaticity shift problems as it alters the transmitted wavelength. To implement  analog dimming, we use  unipolar OOK signal to drive the LED and dimming is achieved by altering the driving current to match the dimming target. In this context, the  driving current, and, consequently, the transmitted optical power, are set to be  proportional to the dimming factor $\gamma_d$. In this case, the  error in decoding the \emph{k}-th signal conditioned on the previous detections can be obtained by (\ref{equ:BER2}), (\ref{equ:BER_noisy2}) and (\ref{equ:BER3}) for  perfect, noisy and outdated CSI, respectively  by changing the transmit power from $P$ to $\gamma_d P$.% whereas the achievable data rate for the \emph{k}-th user becomes
%\begin{equation}
%R_k = B \log_2\left(1+\frac{ (\gamma h_k P_k)^2}{\sum_{j=1}^{k-1} (e_j \gamma h_k P_j)^2 + \sum_{l=k+1}^{N} (\gamma h_k P_l)^2 + \frac{{\sigma_n}^2}{\gamma_d^2} }\right),
%\end{equation}

On the contrary, digital dimming imposes a  pulse width modulation (PWM) signal with a  duty cycle that  is determined by the required dimming factor. This technique alleviates the chromaticity shifts, but at a cost  of reduced spectral efficiency. In the present analysis, we implement digital dimming by means of VOOK as in \cite{dimming3}, where the brightness of the LED is controlled by adopting the data duty cycle $\delta_d$ of the OOK signal. Based on this, information bits are transmitted   when the duty cycle is on, while the off portion is filled with dummy bits that are either zeros or ones, depending on the dimming factor $\gamma_d$. VOOK codewords are depicted in Table  \ref{Tab:VOOK}, indicating that  when  $\gamma_d = 0$, the lights are completely turned off while when $\gamma_d = 1$, full brightness is achieved. Accordingly,  no data bits are transmitted when $\gamma_d$ is set to $0$ or $1$. To this effect, with the use of coded VOOK, the BER in decoding the \emph{k}-th signal conditioned on the previous k-1 detections  can be expressed as follows
\begin{equation}\label{equ:BER_VOOK}
\begin{aligned}
\mbox{Pr}_{e_k VOOK}|{e_1}, \ldots ,{e_{k-1}}= \sum_{\lceil i=n/2 \rceil}^{n}\binom{n}{i} {\mbox{Pr}^i_{e_k|{e_1}, \ldots ,{e_{k-1}} }}   (1-\mbox{Pr}_{e_k|{e_1}, \ldots ,{e_{k-1}} })^{n-i}
\end{aligned}
\end{equation}
where $\mbox{Pr}_{e_k|{e_1}, \ldots ,{e_{k-1}}}$ can be  obtained from (\ref{equ:BER2}), (\ref{equ:BER_noisy2}) and (\ref{equ:BER3}) for  perfect, noisy and outdated CSI respectively, and $n$ is the number of redundant bits in the VOOK codewords, calculated as
\begin{equation}\label{equ:n_VOOK}
\begin{aligned}
n=\begin{cases} 20 \gamma_d, \qquad & 0<\gamma_d\leq  \frac{1}{2}  \\
20 - 20 \gamma_d, \qquad &  \frac{1}{2} \leq\gamma_d<1.
\end{cases}
\end{aligned}
\end{equation}
It is noted that the use of  redundant bits in VOOK leads to BER improvements as $n$ increases, yet,  digital dimming affects the  system spectral efficiency as the achievable data rate  deteriorates lineally with the number of bits in the corresponding  codewords.

\begin{table}[ht]
\small
\caption{VOOK Codewors}
\center
\begin{tabular}{c c c}
\hline
\textbf{$\gamma_d$}&\textbf{$\delta_d$}&\textbf{VOOK Codeword} \\ \hline\hline
1.0&$1.0$& 1111111111\\
0.9&$0.2$& dd11111111\\
0.8&$0.4$& dddd111111\\
0.7&$0.6$& dddddd1111\\
0.6&$0.8$& dddddddd11\\
0.5&$1.0$& dddddddddd\\
0.4&$0.8$& dddddddd00\\
0.3&$0.6$& dddddd0000\\
0.2&$0.4$& dddd000000\\
0.1&$0.2$& dd00000000\\
0.0&$0.0$& 0000000000\\
\hline
\end{tabular}
\label{Tab:VOOK}
\end{table}

\section{Results and Discussion}
\label{sec:results}
In this section, we employ the derived analytic expressions in the analysis of the considered setups. Respective results from extensive Monte Carlo simulations are also provided to verify the validity and usefulness of the offered results. Thus, the  BER performance of a NOMA-VLC downlink system is analyzed for different  scenarios based on the system and channel models  in Section \ref{sec:model}.  It should be noted here that the considered channel model is independent of the room geometry, To this end and without loss of generality, we consider that the users exist within  a $4m \times 4m \times 3m = 48m^3$ room environment.

We consider one transmitting LED mounted at the centre of the ceiling and, in order to ensure comparability, fixed LED transmitting power is used in all  scenarios. We also assume the existence of three users in the coverage area of the transmitting LED that are planned to be served simultaneously using NOMA. It is emphasized here  that this number of users is selected for indicative purposes and that the considered system model is generic and applicable to any number of users. In this context, the  LED superimposes the signals of the three users in the power domain by allocating the power values $P_1$, $P_2$ and $P_3$ to $U_1$, $U_2$ and $U_3$ respectively, under the constraint  $P_1 + P_2 + P_3 = P_{\rm LED}$. In the used notation, $U_i$ denotes  the user in \emph{i}-th decoding order; that is, $h_i$ is in the \emph{i}-th ascending order of the channel gains. The channel gains of all users along with other    system  parameters are depicted in  Table \ref{Tab:Parameters}. For the results that involve users' mobility, the locations of the users are randomly generated. 
It should be noted here that the underlying symmetry in VLC systems may lead to similar channel gains and, consequently, resulting to a  higher error rate. This issue was addressed in an earlier work  \cite{NOMA_VLC}, where we proposed a strategy based on tuning the FOVs of the receiving PDs in order to maximize the differences between the channel gains and improve the performance of NOMA.

\begin{figure*}[   \widetilde{}]
\centering
\normalsize
\subfloat[]\centering{\label{ref_label1}\includegraphics[width=3in,height=2.5in]{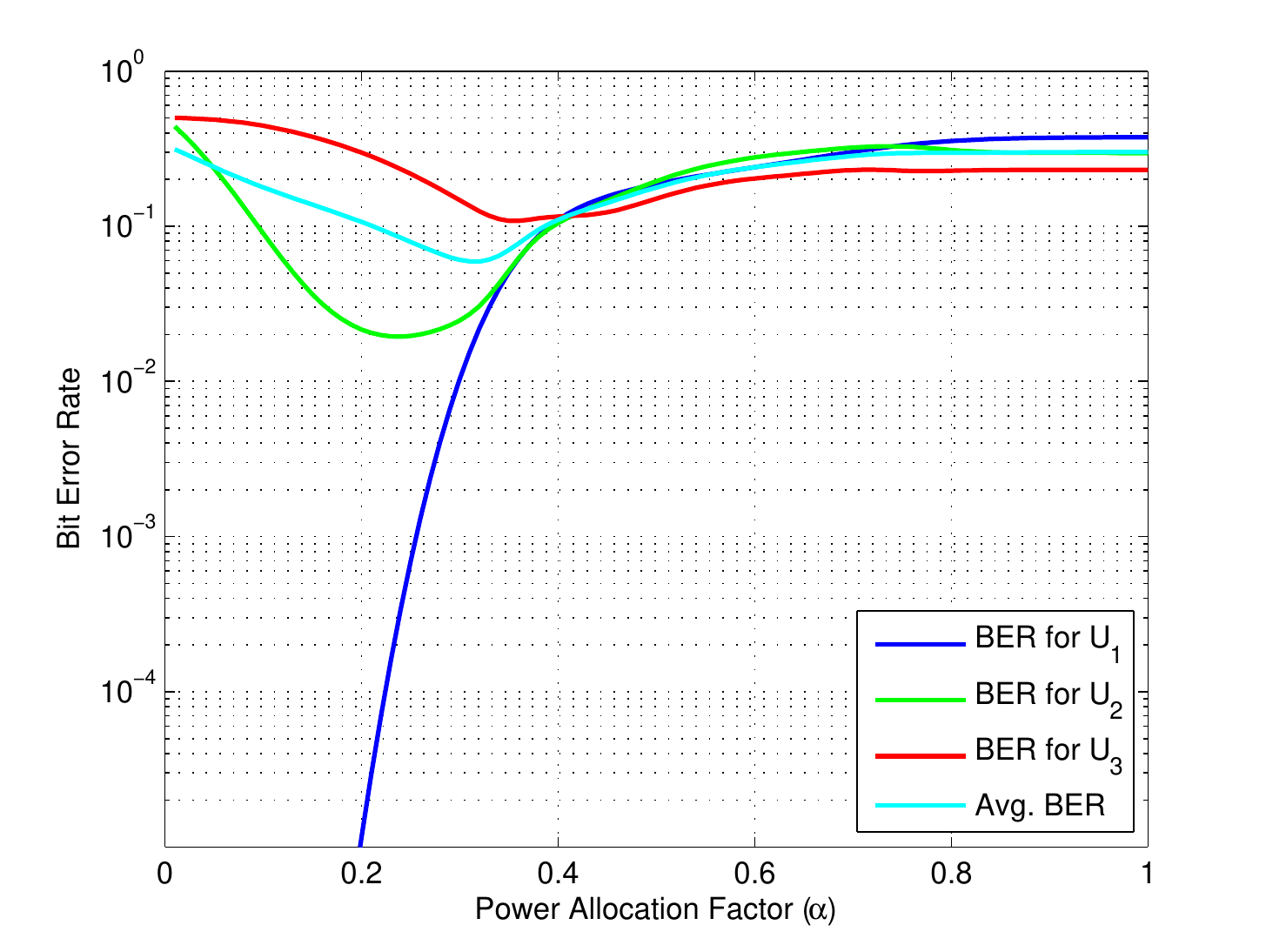}}
\subfloat[]\centering{\label{ref_label2}\includegraphics[width=3in,height=2.5in]{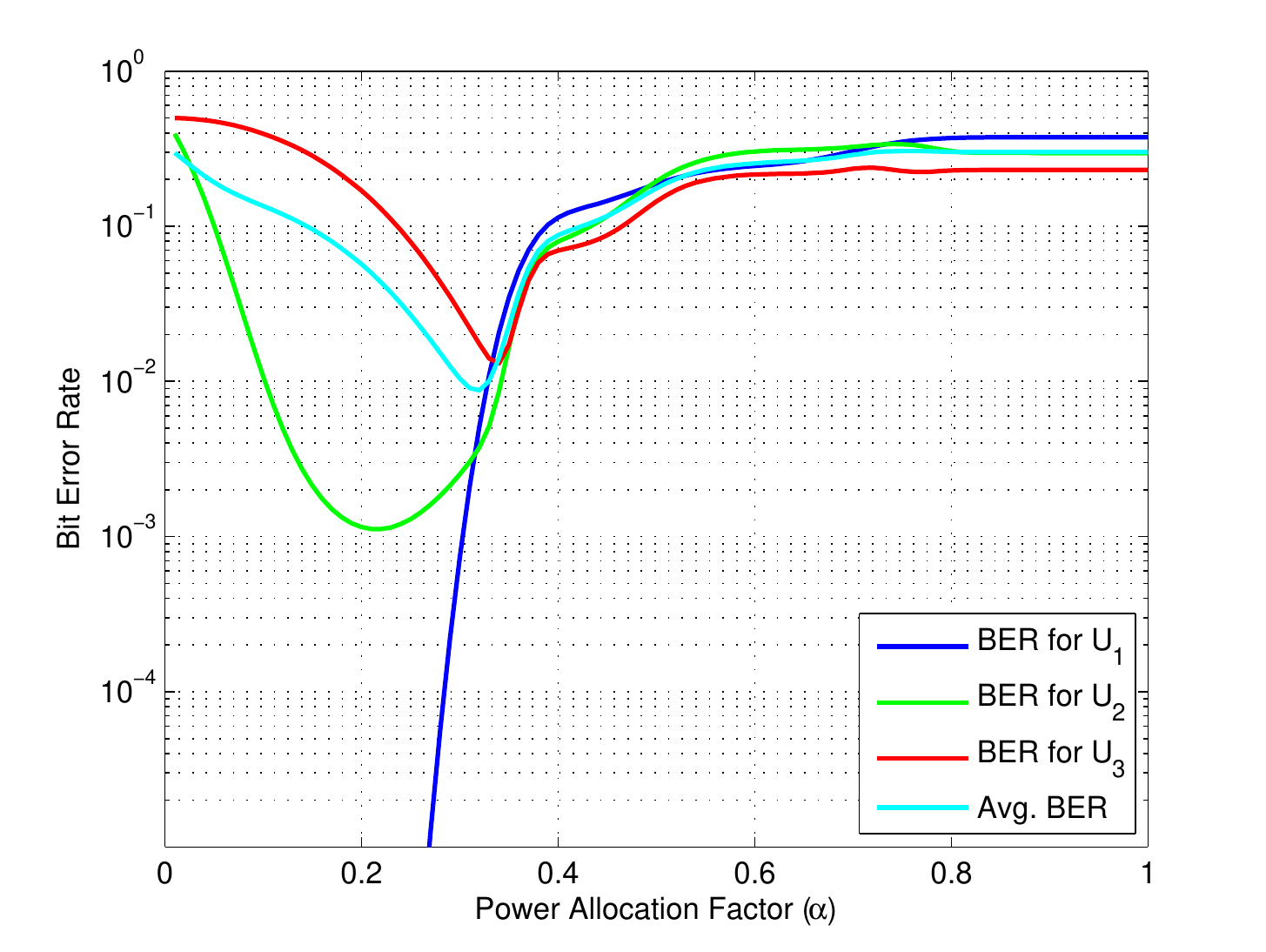}}
\subfloat[]\centering{\label{ref_label2}\includegraphics[width=3in,height=2.5in]{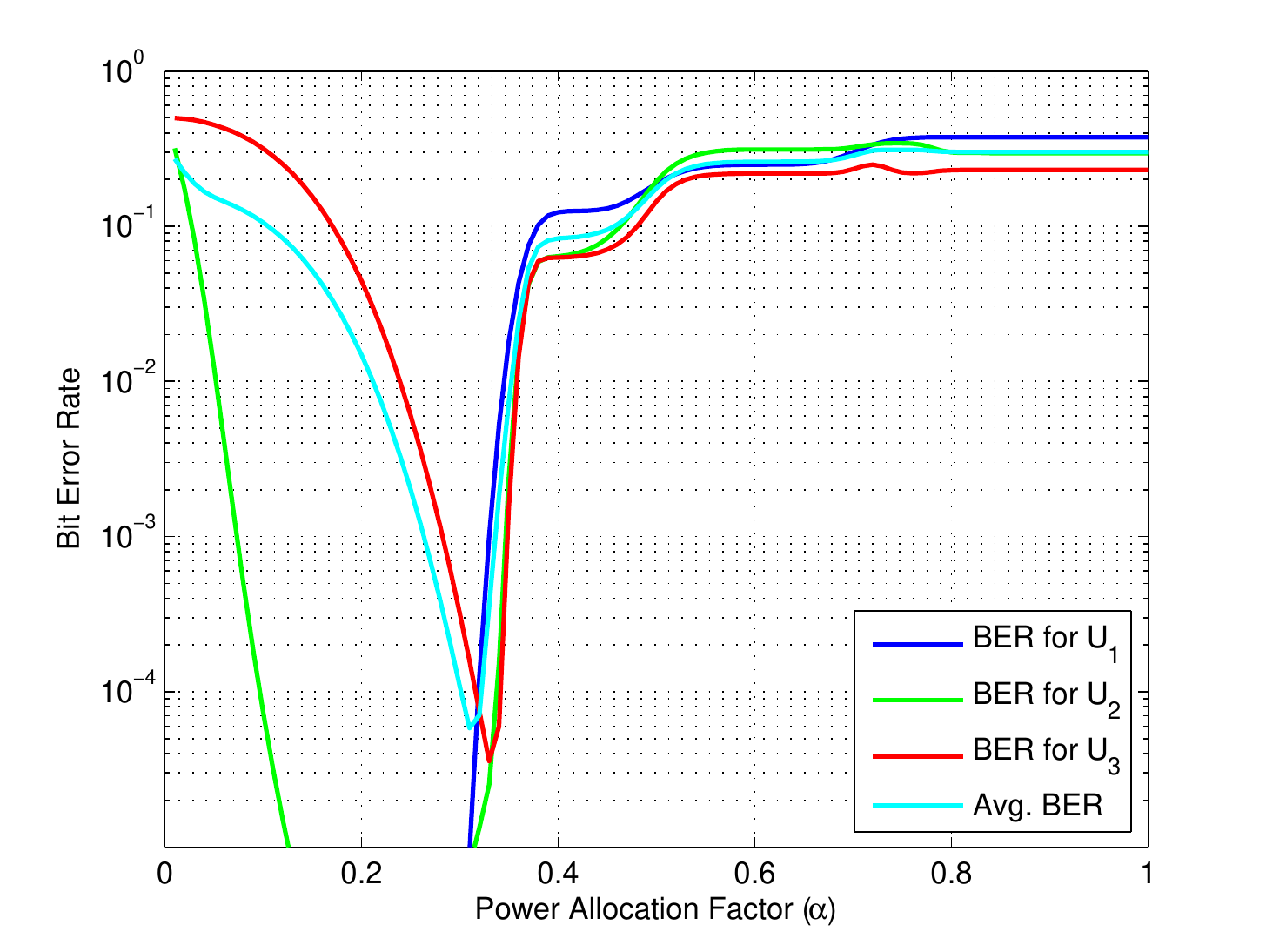}}
\caption{\label{fig:BERvsRho}BER performance under fixed power allocation   (a) SNR=110 dB, (b) SNR=115 dB and (c) SNR=120 dB.}
%\noindent\makebox[\linewidth]{\rule{520pt}{.4pt}}
\end{figure*}

\begin{table}[ht]
\small
\caption{Simulation Parameters}
\center
\begin{tabular}{l l l}
\hline
\textbf{Description}&\textbf{Notation}&\textbf{Value} \\[1 ex] \hline\hline
LED power&$P_{\rm LED}$& 0.25 W\\
Transmitter semi-angle&$\varphi_{i}$& $50$ deg \\
FOV of the PDs&$\phi_{c_i}$& $45$ deg \\
Physical area of PD&$A_i$ &$1.0$ ${\rm cm^2}$ \\
Refractive index of PD lens&$n$&$ 1.5$ \\
Gain of optical filter&$T_s (\phi_{li})$&$ 1.0$ \\
Data rate&$B$&$10$ ${\rm Mbps}$\\
Total number of users&$N$&3 \\
Channel gain of $U_1$ &$h_1$&$0.2835\times10^{-4}$\\
Channel gain of $U_2$ &$h_2$&$0.4787\times10^{-4}$\\
Channel gain of $U_3$ &$h_3$&$0.5272\times10^{-4}$\\
\hline
\end{tabular}
\label{Tab:Parameters}
\end{table}

We evaluate the BER performance with regard to the transmit SNR in order to include the individual path gain of each user. Since the channel gain is in the order of $10^{-4}$, the corresponding results   exhibit an offset of about  $80$ dB  with respect to the  SNR at the receiver site.
First, we investigate the effect of the power allocation factor $\rho$ in (\ref{equ:FPA}) on the BER performance under fixed power allocation. To this end,   Fig. \ref{fig:BERvsRho} shows the  average BER and the individual BER for the three users versus $\rho$ for different transmit SNR values of $110$ dB, $115$ dB and $120$ dB. It is shown that, despite its poor channel conditions, the user in the first decoding order, i.e., $U_1$,   provide comparable error performance to other users. This is because the signal intended for this user is transmitted with a high power compared to other signals, in order to enable $U_1$ to directly decode its signal of interest regardless of the interference it receives. Moreover,
the best average BER performance for all users can be achieved at about $\rho = 0.3$, as in this value, the power levels allocated to users experiencing  low channel gains are sufficiently high to enable correct signal decoding. As $\rho$ increases, users with good channel conditions  receive with higher power levels, which in turn reduces the power associated with the signals of the other users. As a consequence, high errors occur in the early stages of SIC decoding, and are then  inherited  to the following stages, which ultimately leads to poor BER performance. For the rest of the results in this Section, $\rho=0.3$ is considered.

The BER expression derived in Section \ref{sec:analysis}  is validated by Fig. \ref{fig:perfectCSI}, which  shows the BER performance of the three users assuming perfect CSI. It is shown that the derived analytic results are in excellent agreement with the respective Monte Carlo simulation results. It is evident that the user with the lowest decoding order exhibits the best BER performance, while the performance degrades as the decoding order increases. Yet, all users exhibit satisfactory performance above transmit SNR of $120$ dB that corresponds to receive  SNR of about $40$ dB, which is a typical range in VLC transmissions.
\begin{figure}[h]
\center
\includegraphics[width=5in,height=4in]{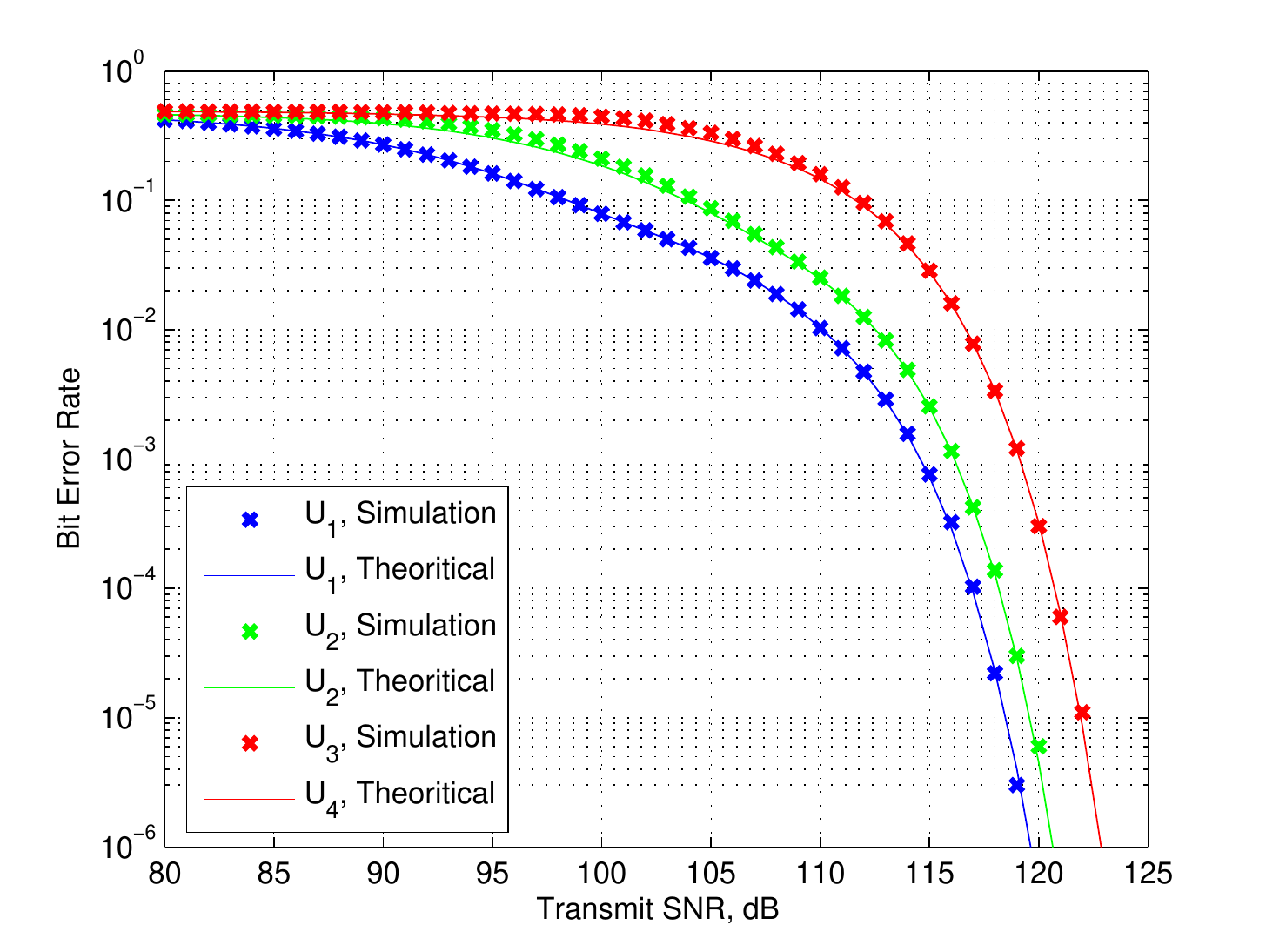}
\caption{BER Performance with Perfect CSI.}
\label{fig:perfectCSI}
\end{figure}
In the following, we investigate the effect of channel uncertainty on the performance of NOMA-VLC. To this end, we assume that the uplink to the LED is error free, so that the LED and the users have the same estimates of the channel gains\footnote{This is a valid assumption for an RF uplink that has been commonly adopted in the literature \cite{CSI1,CSI2}. }. For the case of noisy  CSI, two different models are used: 1) fixed error variance, where   $\sigma_{\epsilon_{n}}^2$ is  independent of the transmit SNR;  2) varying error variance, where $\sigma_{\epsilon_{n}}^2$ is a decreasing function of the transmit SNR. Furthermore, we assume that the noisy channel error variances are identical at the different users. In order to obtain an insight on the impact of different CSI variances on system performance, Fig. \ref{fig:BER_vs_variance} illustrates the BER performance versus different fixed values of $\sigma_{\epsilon_{n}}^2$ at a transmit SNR of $115$ dB. It is clear that users with lower decoding order suffer from higher errors due to the involved  channel uncertainty. This is particularly evident at user $U_1$ that exhibits substantial BER degradation  compared to the error-free CSI (indicated by  dashed line). The reason  is that $U_1$ has the lowest channel gain among all users, which renders signal detection highly sensitive to errors in the available CSI. Moreover, $U_1$ needs to decode $x_1$ with the existence of high interference from the signals of other users, which increases the severity of the effect of imperfect CSI. Furthermore, although other users  also need to decode $x_1$ despite the involved interference, their relatively high channel gains make the detection more robust to channel errors from the early stages of the detection. This is specifically clear for $U_3$ that is the least affected by channel uncertainty.   It should be noted here that the CSI error resulting from the noisy channel is in general  small enough not to affect the ordering of channel gains. Hence, the power allocation based on CSI available at the transmitter is not ultimately  affected. Fig. \ref{fig:noisy_variances} demonstrates  the corresponding  BER performance under fixed and SNR-dependant error variances. It is observed that fixed $\sigma_{\epsilon_{n}}^2$ results to an irreducible error floor at high SNRs for users with low decoding order. However, when $\sigma_{\epsilon_{n}}^2$ is modeled as SNR-dependant, the BER decreases with the increase in transmit SNR. It is not surprising to observe that the performance of user $U_3$ is almost the same for the two variance models,  which is due to the fact that  the impact of error is already insignificant at this user. In order to validate the BER expression for noisy CSI  in \emph{Proposition 1}, we plot the BER performance with the aid of  the derived approximation along with respective results from  Monte Carlo simulations in Fig. \ref{fig:noisyCSI}. It is observed that, the derived  approximation provides accurate results that are in tight agreement with the simulation results. Moreover, it is noted that the user in the first decoding order suffers higher performance degradation compared to users with higher decoding orders. This is due to the fact that $U_1$ does not perform SIC, which means that it has to deal with the existing interference along with the CSI error. Moreover, $U_1$ has the lowest channel gain among all users, which makes the effect of noisy CSI rather detrimental.

\begin{figure}[h]
\center
\includegraphics[width=5in,height=4in]{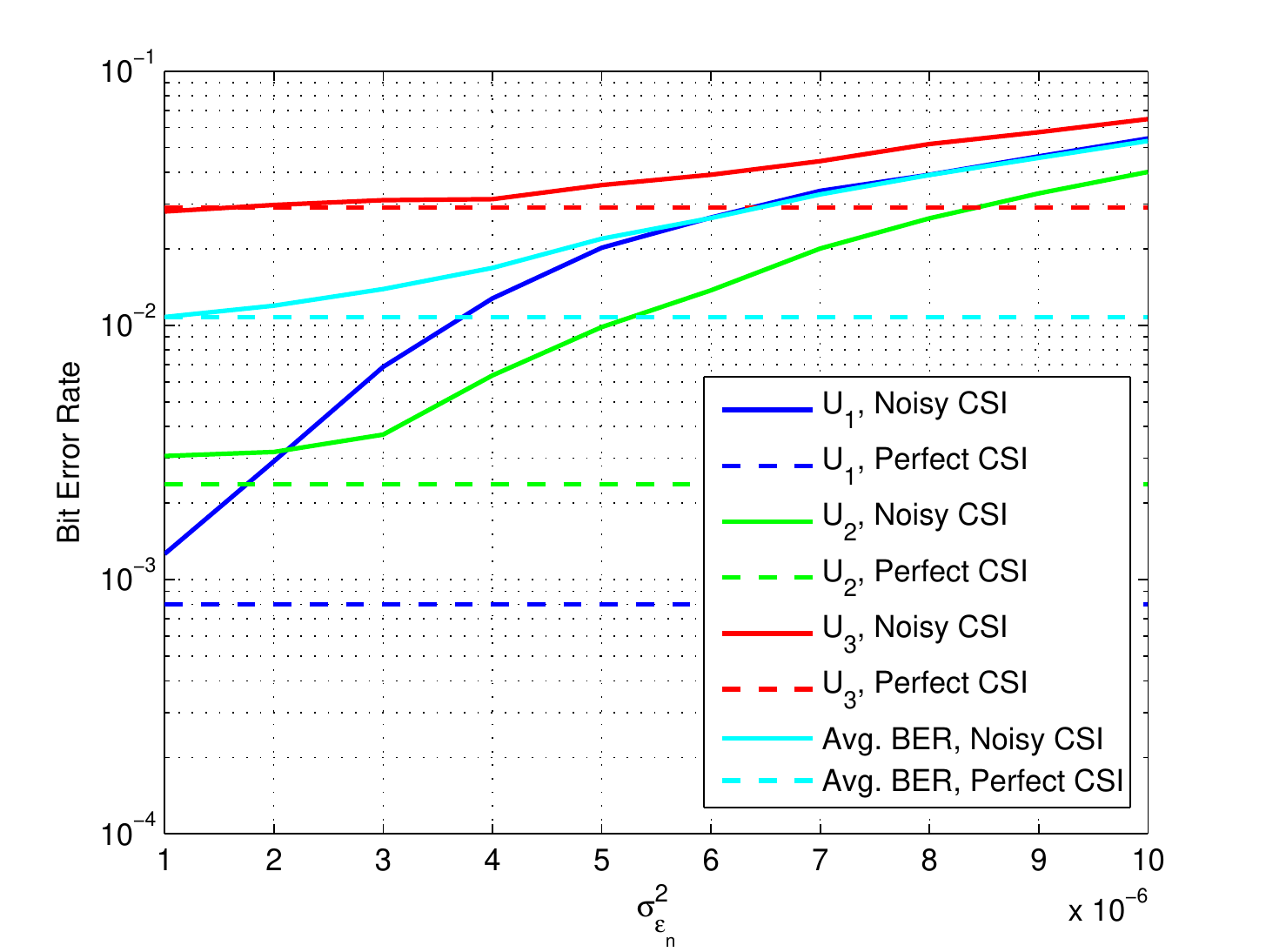}
\caption{BER Performance for different values of $\sigma_{\epsilon_{n}}^2$.}
\label{fig:BER_vs_variance}
\end{figure}

\begin{figure}[h]
\center
\includegraphics[width=5in,height=4in]{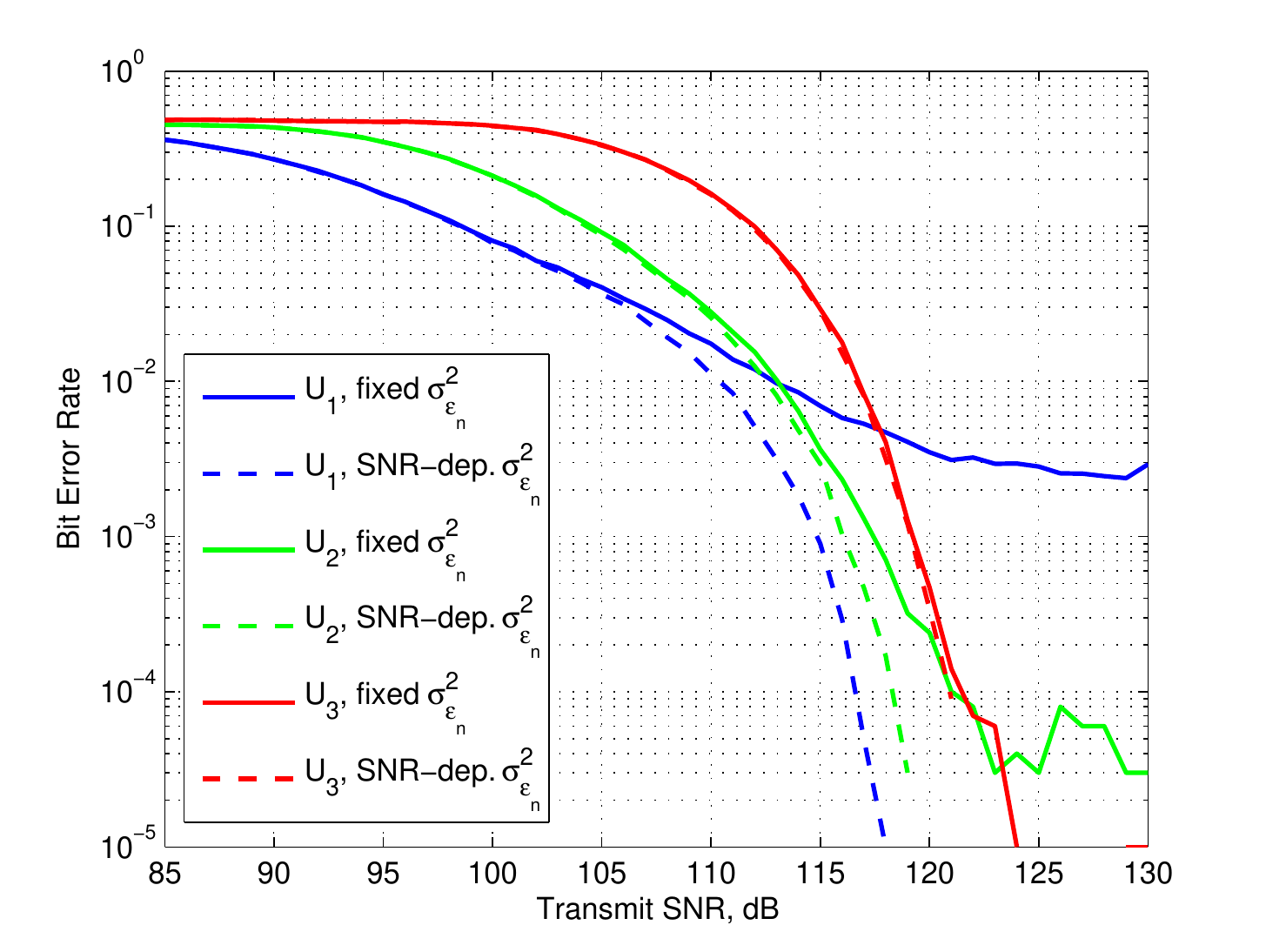}
\caption{BER Performance with Noisy CSI, for fixed and SNR-dependent error variances.}
\label{fig:noisy_variances}
\end{figure}

\begin{figure}[h]
\center
\includegraphics[width=5in,height=4in]{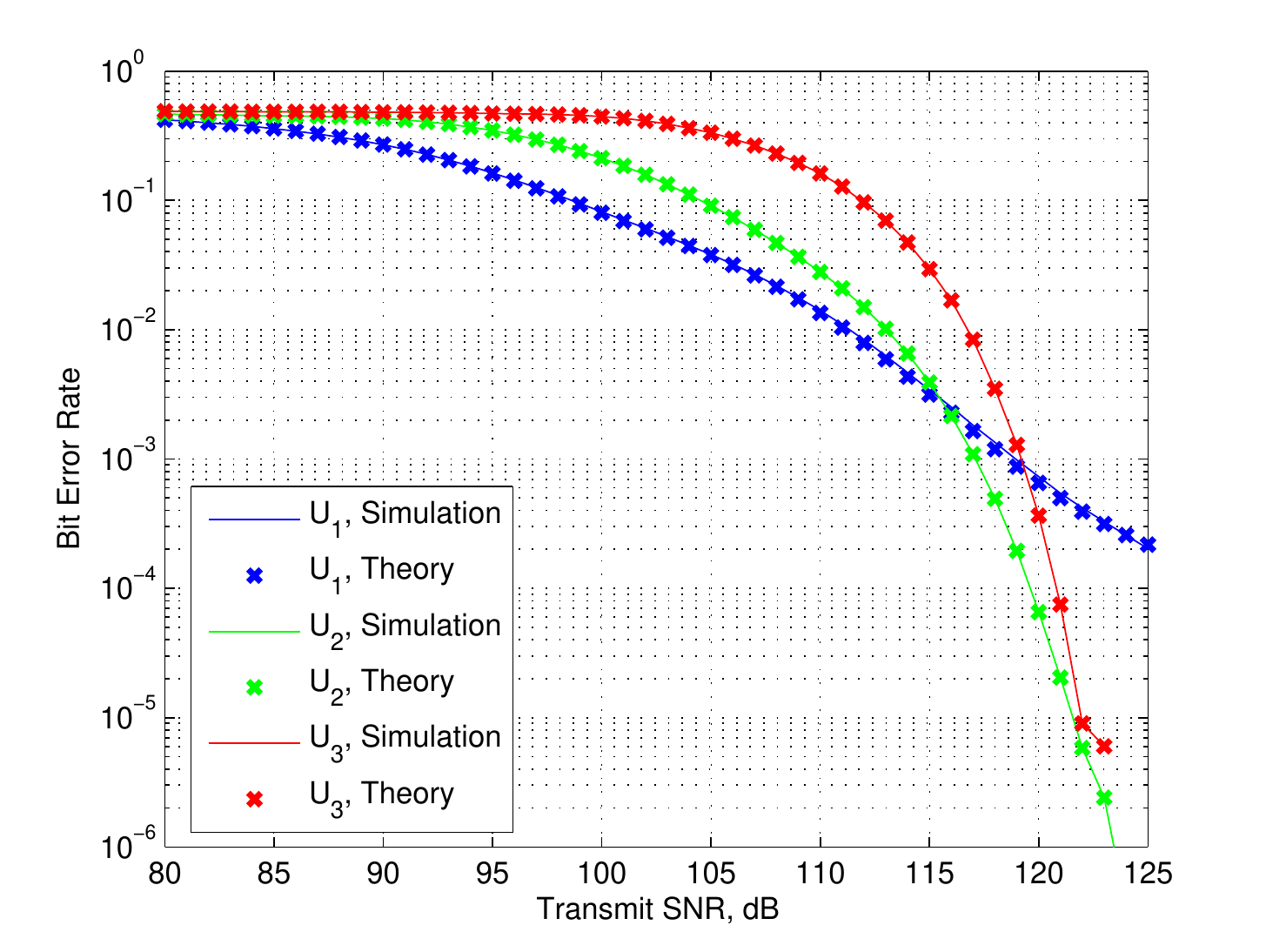}
\caption{BER Performance with Noisy CSI, $\sigma_{\epsilon_{n}}^2 = 2\times10^{-6}$ .}
\label{fig:noisyCSI}
\end{figure}

Next, we quantify  the effect of the mobility of the indoor users. To this end, it is recalled that VLC channels are  mainly dependant upon the user location  with respect to the transmitting LED. As a result, even a slight change in the user's location results in a change of the  corresponding channel gain. To this effect, if a user possesses  an outdated channel estimation, i.e., a change of location occurs before the next channel update, CSI becomes erroneous. In order to quantify  the impact of outdated CSI on the overall system performance, we simulate the mobility of  indoor users  with random speed from $0$ to $2$ m/s while they remain  connected to the same LED. We then assume that the change in location may occur between CSI updates and thus,  both transmitter and receiver  use the outdated CSI for power allocation and decoding, respectively. It is also noted  here that outdated CSI, unlike noisy CSI, may lead to a change in the ordering of the channel gains of the users which ultimately leads  to unfair power allocation at the transmitter, where high power values may be allocated to users  with good channel conditions and vice versa. This results to a dramatic  performance degradation  for users with poor channel gains, as their allocated power becomes insufficient for successful decoding.  In the same context, Fig. \ref{fig:outdated} and Fig. \ref{fig:outdated2} demonstrate the BER performance with outdated CSI for the three users, where the upper bound for the error is determined by (\ref{equ:BER3}), in \emph{Proposition 2}, when the user moves with maximum velocity. In Fig. \ref{fig:outdated}, we simulate the mobility of users assuming that their relative channel ordering remains constant, which  is valid if users change locations in the same trend, i.e. Reference Point Group Model \cite{group_mobility}, which is realistic in large indoor environments, such as museums or airports. On the contrary, Fig.  \ref{fig:outdated2}  illustrates the performance degradation caused by unfair power allocation when outdated CSI leads to change in the ordering of  users' gains. It is noted that for the sake of consistency with other figures, $U_1$ here  denotes  the user that has the lowest channel gain among all users. However, $U_1$ now is not in the first decoding order as its channel is erroneously  estimated not to be the lowest. Therefore, users with low channel gains suffer, as expected,  from dramatic degradation, while $U_3$ benefits from the high power that is in fact erroneously   allocated to it. As a consequence, the  high power allows $U_3$ to detect the desired signal effectively, even though  the estimate of $h$ available at the decoder is practically inaccurate.

\begin{figure}[h]
\center
\includegraphics[width=5in,height=4in]{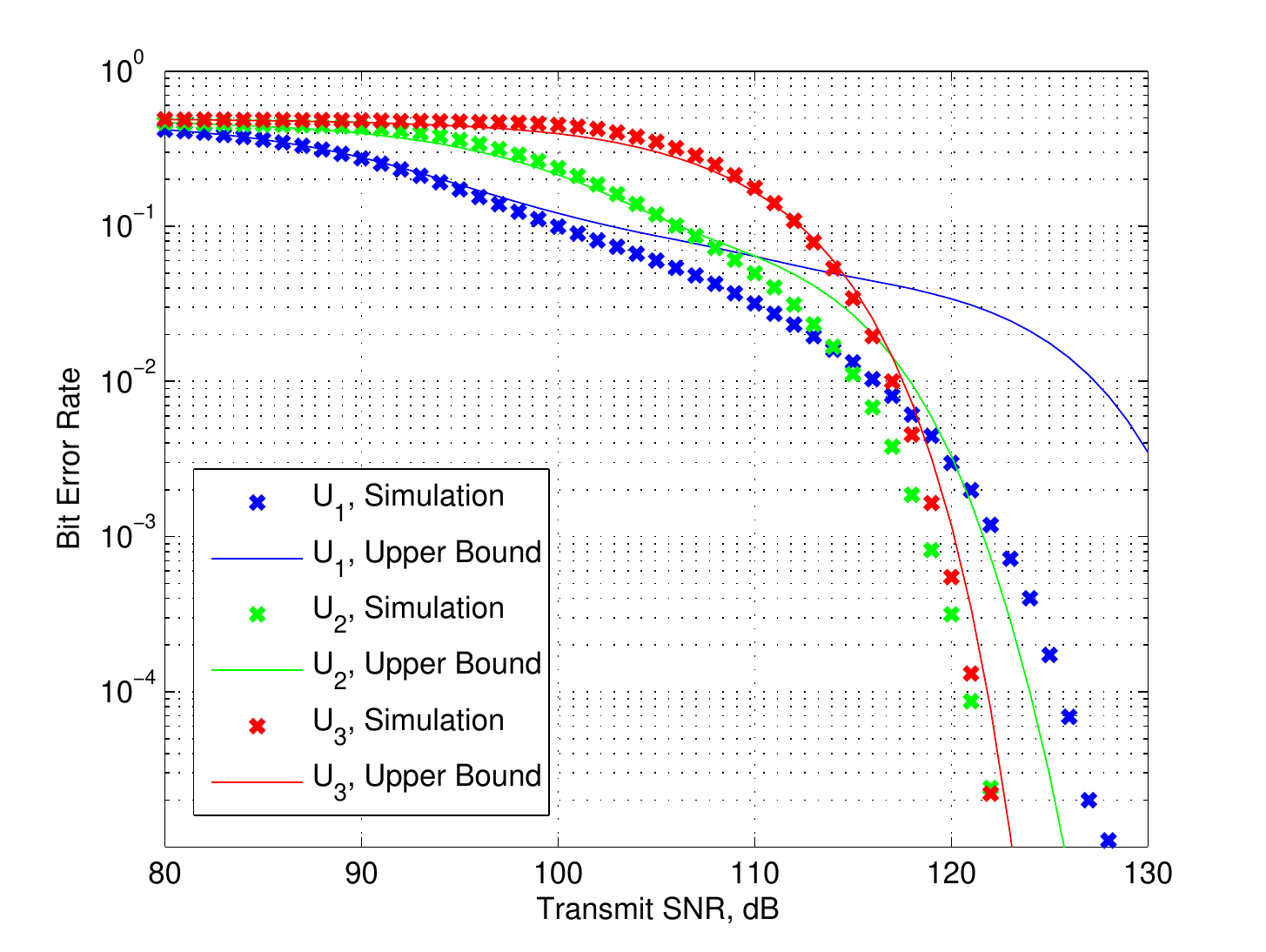}
\caption{BER Performance with Outdated CSI, order not changed.}
\label{fig:outdated}
\end{figure}

\begin{figure}[h]
\center
\includegraphics[width=5in,height=4in]{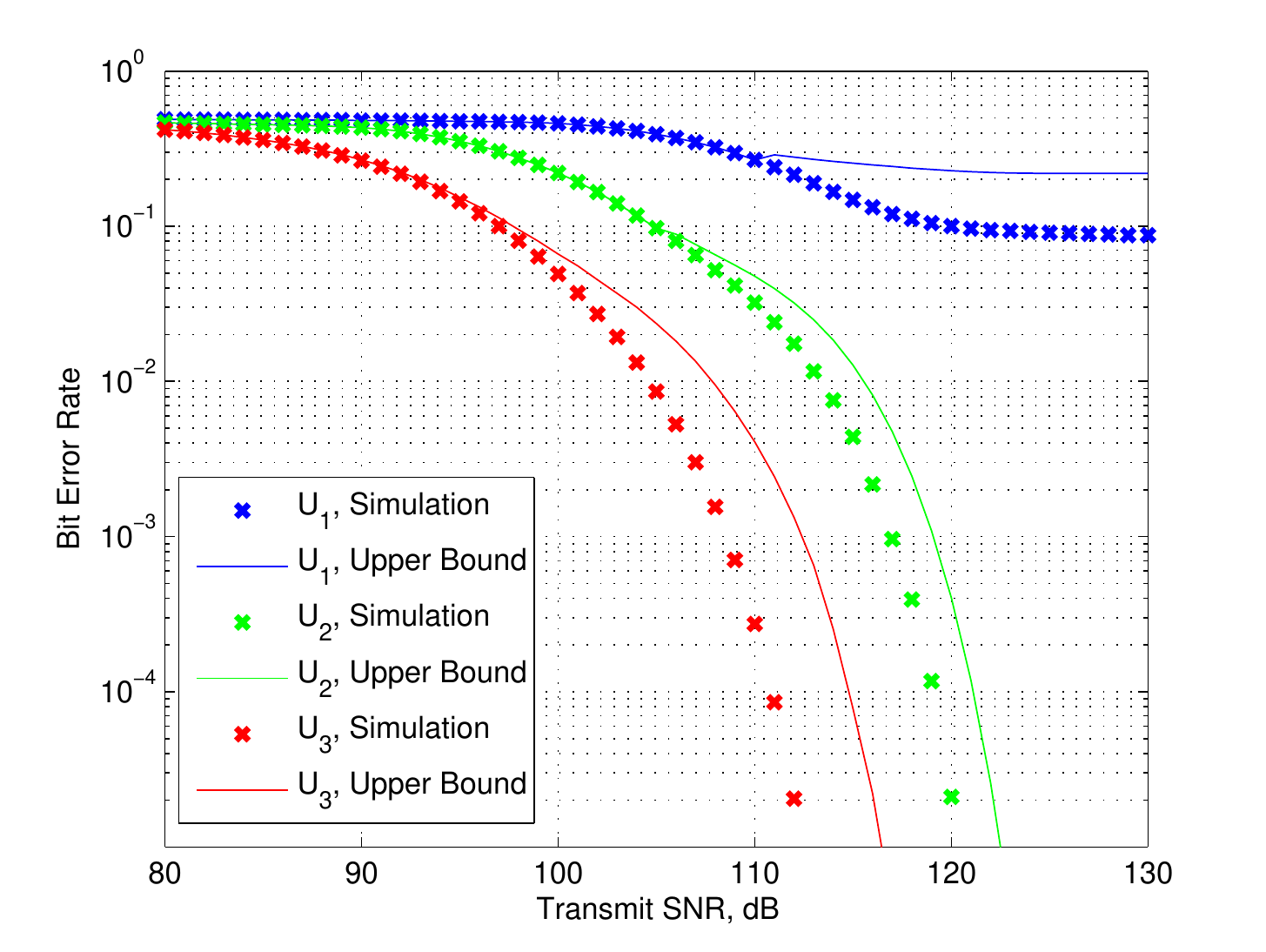}
\caption{BER Performance with Outdated CSI, order changed.}
\label{fig:outdated2}
\end{figure}

Finally, we evaluate  the BER performance of the NOMA-VLC system under dimming control. Fig. \ref{fig:analog} and Fig. \ref{fig:digital} demonstrate the BER of the three users under analog and digital dimming schemes, respectively. As expected, analog intensity dimming leads to    BER performance  degradation, particularly at low $\gamma_d$ values, as lower transmit power is used, which reduces the received SNR. On the contrary, digital dimming employed by means of VOOK leads to BER enhancement, that is substantial when the dimming factor is around $0.5$. This is achieved thanks to  the increase of redundant data bits, which in turn  reduces the probability of incorrect detections. Yet,  this enhancement comes at the cost of reduction in the achievable throughput, which is the main drive for NOMA.    The effect of imperfect CSI under dimming is also demonstrated, which indicates that outdated CSI leads to higher performance degradation compared to noisy CSI, which is in agreement  with the previous results.

\begin{figure}[h]
\center
\includegraphics[width=5in,height=4in]{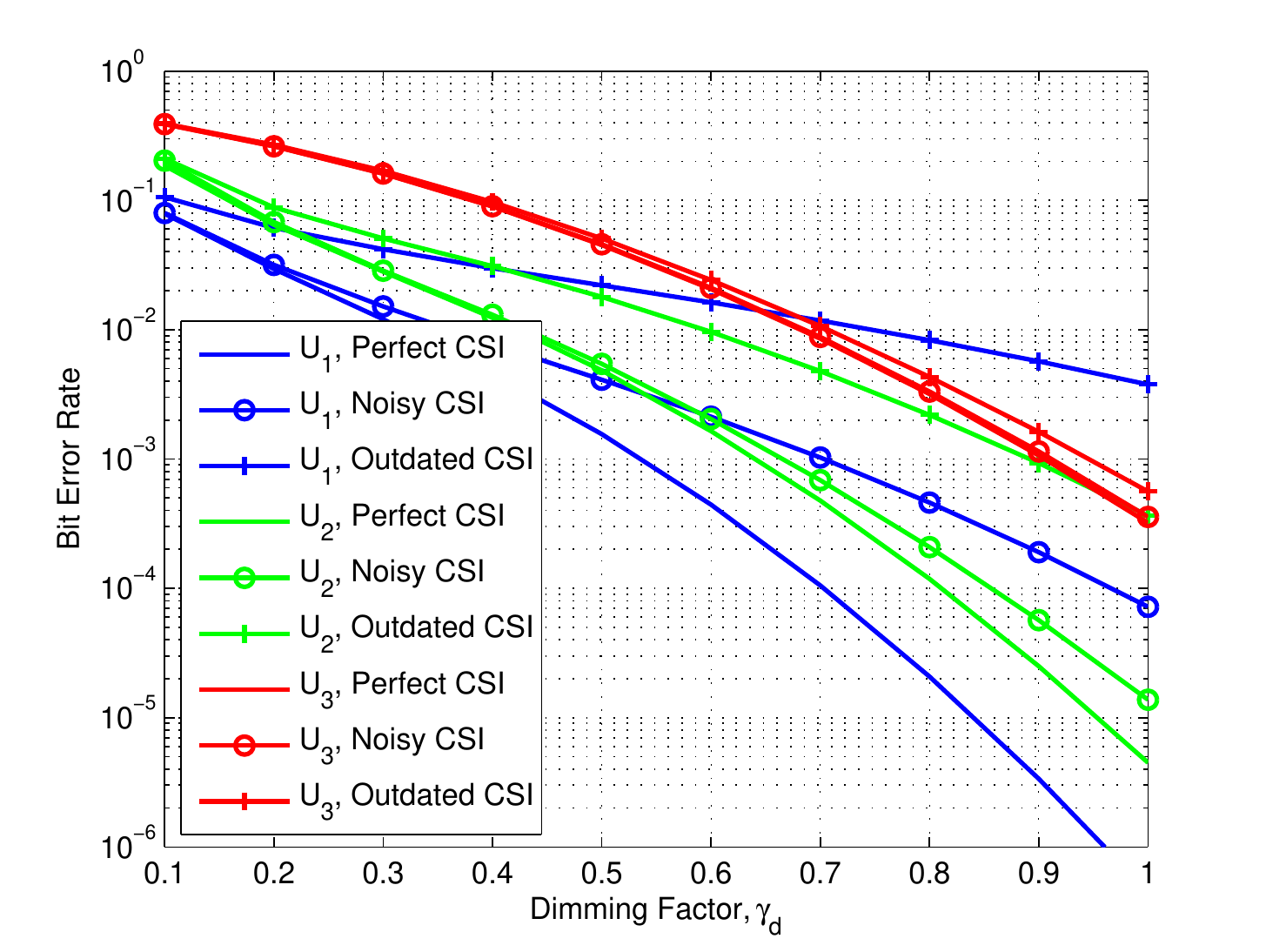}
\caption{BER Performance with analog intensity dimming.}
\label{fig:analog}
\end{figure}

\begin{figure}[H]
\center
\includegraphics[width=5in,height=4in]{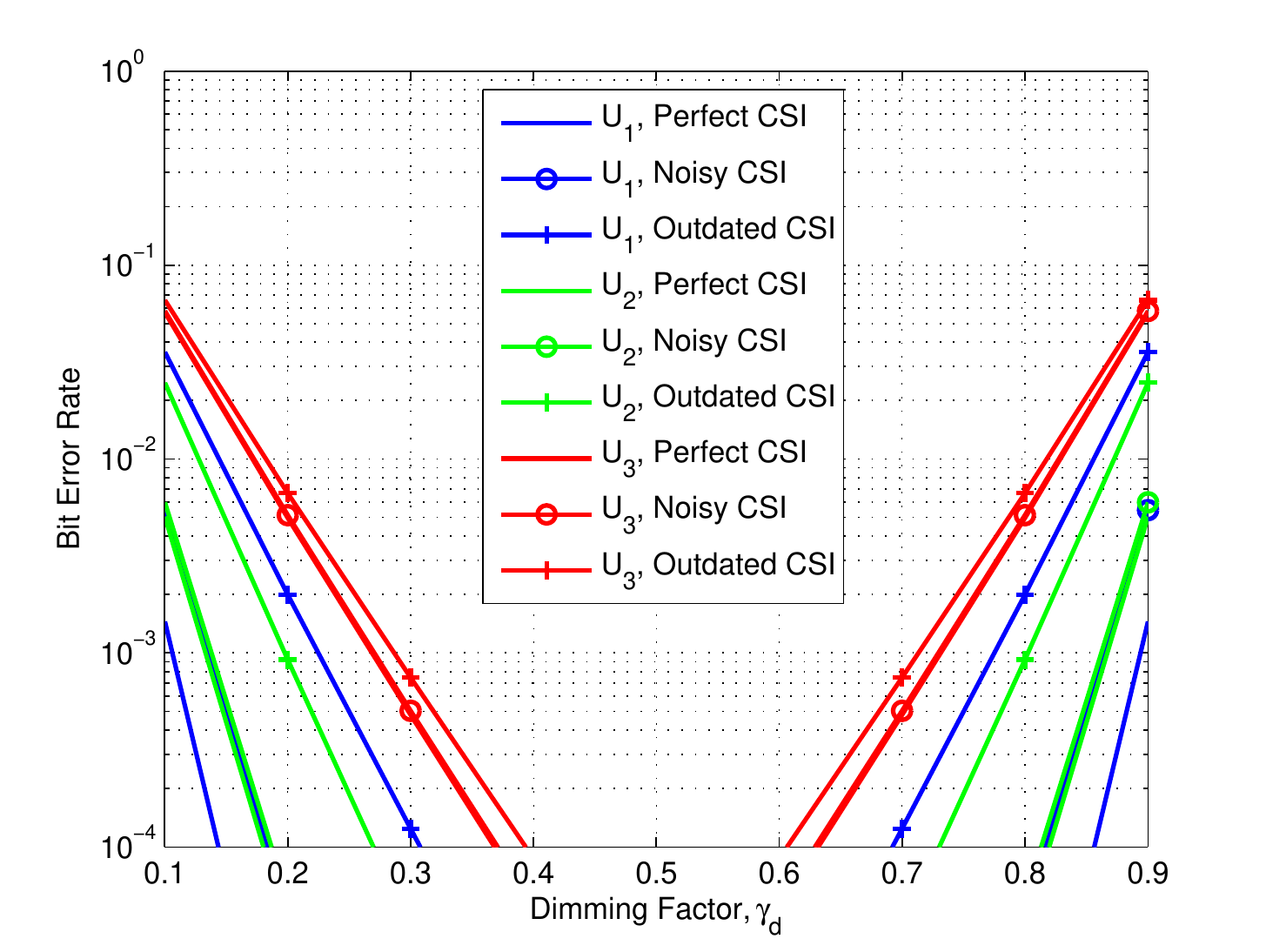}
\caption{BER Performance with digital VOOK dimming.}
\label{fig:digital}
\end{figure}
\section{Conclusions}
\label{sec:conc}
This work was devoted to the analysis of  the BER performance in a downlink VLC network where multiple access was provided by means of NOMA. This was realized for the case of both perfect and imperfect CSI, which  showed that  noisy CSI leads to, as expected,  a  degradation of the system performance. However, this degradation is rather smaller compared to the one created by  outdated CSI which results from the mobility of the user terminal between two CSI updates and cause detrimental performance loss if the ordering of the users' channel gains change between the channel updates. The validity of the derived analytic results was justified  by extensive comparisons with results from respective Monte Carlo simulations. Finally, the offered results provided meaningful insights that are expected to be useful in future design and deployments of VLC systems.

\begin{appendices}

\section{Proof of Theorem 1} \label{app:A}
Using maximum-likelihood (ML) detector, the decoder at the \emph{k}-th receiver decides for the  vector $\hat{s}$  that  minimizes the Euclidean distance between the  received signal vector $y$ and the potential received signals leading to
\begin{equation}\label{equ:ML}
\hat{{s}}= \mbox{arg} \min\limits_{s} \left|{y -  \gamma h_k s}\right|^2.
\end{equation}
Based on this and assuming  that the $U_k$ user cancels  successfully the signals $s_1$ , \ldots , $s_{m-1}$, the error probability at $U_k$ in detecting the signal $s_m$ intended to  user $U_m$ $(1 \leq m \leq  k-1)$ can be expressed as follows: when the transmitted symbol $s_m=0$, the conditional error probability is given by
\begin{equation}\label{equ:perfect0}
\mbox{Pr}_{e_{m\rightarrow k| s_m=0}} = \int_{\frac{1}{2}  \gamma h_k P_m}^{\infty} \mathcal{N}_{y_k}\left( \gamma h_k  \sum \limits_{i=m+1}^{N}  P_i s_i, \sigma_n\right) dy_k
\end{equation}
which can be expressed in closed-form in terms of the $Q-$function, namely
\begin{equation}\label{equ:perfect1}
\mbox{Pr}_{e_{m\rightarrow k| s_m=0}} = \mathcal{Q}\left(\frac{ \gamma h_k}{\sigma_n} \left(\frac{P_m}{2}  - \sum \limits_{i=m+1}^{N}  P_i s_i\right)\right).
\end{equation}
On the contrary, when   $s_m=1$ is transmitted, it follows that
\begin{equation}\label{equ:perfect2}
\mbox{Pr}_{e_{m\rightarrow k| s_m=1}} = \int_{-\infty}^{\frac{1}{2}  \gamma h_k P_m}  \mathcal{N}_{y_k}\left(\gamma h_k \left( P_m + \sum \limits_{i=m+1}^{N}  P_i s_i \right),\sigma_n\right)dy_k
\end{equation}
which after some algebraic manipulations, it can be expressed by the following  closed-form expression
\begin{equation}
\mbox{Pr}_{e_{m\rightarrow k| s_m=1}} = 1- \mathcal{Q}\left(\frac{ \gamma h_k}{\sigma_n} \left(-\frac{P_m}{2} - \sum \limits_{i=m+1}^{N}   P_i s_i\right)\right)
\end{equation}
which with the aid of the identity $Q(-x) = 1 - Q(x)$  can be equivalently re-written as follows:
\begin{equation}
\mbox{Pr}_{e_{m\rightarrow k| s_m=1}} =  \mathcal{Q}\left(\frac{ \gamma h_k}{\sigma_n} \left(\frac{P_m}{2} + \sum \limits_{i=m+1}^{N}   P_i s_i\right)\right).
\end{equation}

%\begin{equation}\label{equ:Q-func0}
%\begin{aligned}
%\begin{cases}
%\mbox{Pr}_{e_{m\rightarrow k| s_m=0}} &= \int_{\frac{1}{2}  \gamma h_k P_m}^{\infty} \mathcal{N}_{y_k}\left( \gamma h_k  \sum \limits_{i=m+1}^{N}  P_i s_i, \sigma_n\right) dy_k\\
%&= \mathcal{Q}\left(\frac{ \gamma h_k}{\sigma_n} \left(\frac{P_m}{2}  - \sum \limits_{i=m+1}^{N}  P_i s_i\right)\right)\\
%\mbox{Pr}_{e_{m\rightarrow k| s_m=1}} &= \int_{-\infty}^{\frac{1}{2}  \gamma h_k P_m}  \mathcal{N}_{y_k}\left(\gamma h_k \left( P_m + \sum \limits_{i=m+1}^{N}  P_i s_i \right),\sigma_n\right)dy_k\\
%&=1- \mathcal{Q}\left(\frac{ \gamma h_k}{\sigma_n} \left(-\frac{P_m}{2} - \sum \limits_{i=m+1}^{N}   P_i s_i\right)\right).
%\end{cases}
%\end{aligned}
%\end{equation}
It is noted that the above expressions assume perfect cancelation of the first $m-1$ signals. Nevertheless, detection errors may practically occur in any step  of the successive cancelation process. Therefore, considering the contribution of the residual interference inherited by cancelation errors,    equations (\ref{equ:perfect1})-(\ref{equ:perfect2}) can be alternatively re-written as
\begin{equation}\label{equ:perfect3}
\begin{aligned}
\mbox{Pr}_{e_{m\rightarrow k| s_m=0}} &= \int_{\frac{1}{2}  \gamma h_k P_m}^{\infty}  \mathcal{N}_{y_k}\left( \gamma h_k \left( \sum_{j=1}^{m-1}  {e_j} P_j + \sum \limits_{i=m+1}^{N}   P_i s_i \right),\sigma_n\right) dy_k\\
&= \mathcal{Q}\left(\frac{ \gamma h_k}{\sigma_n} \left(\frac{P_m}{2}  - \sum_{j=1}^{m-1}   {e_j} P_j - \sum \limits_{i=m+1}^{N}  P_i s_i\right)\right)
\end{aligned}
\end{equation}
and
\begin{equation}\label{equ:perfect4}
\begin{aligned}
\mbox{Pr}_{e_{m\rightarrow k| s_m=1}}& = \int_{-\infty}^{\frac{1}{2} \gamma h_k P_m }  \mathcal{N}_{y_k}\left( \gamma h_k \left(P_m + \sum_{j=1}^{m-1}   {e_j} P_j + \sum \limits_{i=m+1}^{N}  P_i s_i \right) ,\sigma_n\right) dy_k\\
&= \mathcal{Q}\left(\frac{ \gamma h_k}{\sigma_n} \left( \frac{P_m}{2}  + \sum_{j=1}^{m-1}  {e_j} P_j + \sum \limits_{i=m+1}^{N}   P_i s_i\right)\right)
\end{aligned}
\end{equation}
%\begin{equation}\label{equ:Q-func2}
%\begin{aligned}
%\begin{cases}
%\mbox{Pr}_{e_{m\rightarrow k| s_m=0}}&= \int_{\frac{1}{2}  \gamma h_k P_m}^{\infty}  \mathcal{N}_{y_k}\left( \gamma h_k \left( \sum_{j=1}^{m-1}  {e_j}_\iota P_j + \sum \limits_{i=m+1}^{N}   P_i s_i \right),\sigma_n\right) dy_k\\
%&= \mathcal{Q}\left(\frac{ \gamma h_k}{\sigma_n} \left(\frac{P_m}{2}  - \sum_{j=1}^{m-1}   {e_j}_\iota P_j - \sum \limits_{i=m+1}^{N}  P_i s_i\right)\right)\\
%\mbox{Pr}_{e_{m\rightarrow k| s_m=1}} &= \int_{-\infty}^{\frac{1}{2} \gamma h_k P_m }  \mathcal{N}_{y_k}\left( \gamma h_k \left(P_m + \sum_{j=1}^{m-1}   {e_j}_\iota P_j + \sum \limits_{i=m+1}^{N}  P_i s_i \right) ,\sigma_n\right) dy_k\\
%&=1 -  \mathcal{Q}\left(\frac{ \gamma h_k}{\sigma_n} \left(- \frac{P_m}{2}  - \sum_{j=1}^{m-1}  {e_j}_\iota P_j - \sum \limits_{i=m+1}^{N}   P_i s_i\right)\right).
%\end{cases}
%\end{aligned}
%\end{equation}
respectively. Based on the above,  the total error probability in decoding $s_k$ at $U_k$ can be finally  obtained by summing up all conditional error probabilities of the previous detections, which completes the proof.

\section{Proof of Proposition 1} \label{app:B}
According to (\ref{equ:noisyCSI}), the ML decision rule at user $U_k$ is readily expressed as
\begin{equation}\label{equ:ML}
\hat{{s}}= \mbox{arg} \min\limits_{s} \left|{y - \gamma \hat{h}_k s}\right|^2.
\end{equation}
To this effect,  the error probability at $U_k$ in detecting the signal $s_m$ intended to  user $U_m$ $(1 \leq m \leq  k-1)$ can be represented as follows:
\begin{equation}\label{equ:noisy1}
\mbox{Pr}_{e_{m\rightarrow k| s_k=0}}= \int_{\frac{1}{2} \gamma P_m \hat{h}_k}^{\infty}  \mathcal{N}_{y_k}\left( \gamma h_k\left(\sum_{j=1}^{m-1}  {e_j} P_j + \sum \limits_{i=m+1}^{N}   P_i s_i\right),\sigma_n\right) dy_k
\end{equation}
and
\begin{equation}\label{equ:noisy2}
\mbox{Pr}_{e_{m\rightarrow k| s_k=1}} = \int_{-\infty}^{\frac{1}{2} \gamma P_m \hat{h}_k}  \mathcal{N}_{y_k}\left( \gamma h_k\left(P_m + \sum_{j=1}^{m-1}   {e_j} P_j + \sum \limits_{i=m+1}^{N}  P_i s_i\right),\sigma_n\right) dy_k
\end{equation}
%\begin{equation}\label{equ:Q-func2-noisy}
%\begin{aligned}
%\begin{cases}
%&\mbox{Pr}_{e_{m\rightarrow k| s_k=1}}\\ &= \frac{1}{\sqrt{2\pi \sigma_n^2}} \int_{\frac{1}{2}P_m \hat{h}_k}^{\infty} e^{-\frac{\left(y_k - h_k \left( \sum_{j=1}^{m-1}  e_j P_j + \sum \limits_{i=m+1}^{N}   P_i s_i \right) \right)^2}{2 \sigma_n^2}}\\
%&\mbox{Pr}_{e_{m\rightarrow k| s_k=1}} \\& = \frac{1}{\sqrt{2 \pi \sigma_n^2}} \int_{-\infty}^{\frac{1}{2}P_m \hat{h}_k} e^{-\frac{\left(y_k - h_k \left(P_m + \sum_{j=1}^{m-1}   e_j P_j + \sum \limits_{i=m+1}^{N}  P_i s_i \right) \right)^2}{2 \sigma_n^2}},
%\end{cases}
%\end{aligned}
%\end{equation}
where ${e_j} = \hat{s_j} - s_j$. Based on this, and after some algebraic manipulations,   the conditional error probability in (\ref{equ:BER2}) becomes
\begin{equation}
\label{equ:BER_noisy}
\begin{aligned}
\mbox{Pr}_{e_k}|{e_1}, \ldots ,{e_{k-1}}   &=  \frac{1}{2^{N-k+1}} \int_{- \infty}^{\infty} \Bigg( \sum_{i=1}^{2^{N-k}} \mathcal{Q}\bigg(\frac{ \gamma P_k \epsilon_{n}}{2\sigma_n} +\frac{ \gamma h_k}{\sigma_n} (\frac{P_k}{2} - \sum_{j=1}^{k-1} {e_j} P_j - \sum_{l=k+1}^{N} P_l A_{il})\bigg)
 \\ &  +   \sum_{i=1}^{2^{N-k}} \mathcal{Q}\bigg(\frac{ -\gamma P_k \epsilon_{n}}{2\sigma_n} +\frac{ \gamma h_k}{\sigma_n} (\frac{P_k}{2} + \sum_{j=1}^{k-1} {e_j} P_j + \sum_{l=k+1}^{N} P_l A_{il}) \bigg)  \Bigg) \times  \mathcal{N}_{\epsilon_{n}}(0,\sigma_{\epsilon_{n}}^2)dh_k.
 \end{aligned}
\end{equation}

It is evident that the derivation of an exact  closed-form expression to (\ref{equ:BER_noisy}) is subject to analytic evaluation of the involved two integrals. However, this is unfortunately not feasible as these integrals are not available in the open technical literature and in  tabulated form. Yet, a relatively simple closed-form approximation can be derived instead, which appears to be particularly accurate for all values of the considered scenario. To this end, it is recalled that the one dimensional Gaussian $Q$-function can be accurately expressed by an accurate approximation in \cite{Q-function}, namely
\begin{equation}\label{equ:q_function}
\begin{aligned}
&\mathcal{Q}(x)\approx e^{a x^2 + b x + c},  &   \qquad    x \geq 0
\end{aligned}
\end{equation}
where $a, b, c \in \mathbb{R}$ are the corresponding  fitting parameters that are selected according to the fitting criteria.
These values are available in \cite{Q-function} and ensure increased tightness as the corresponding involved absolute and relative errors between the exact and approximated values are particularly  small for the entire range of values of $x$. To this effect, by performing the necessary variable transformation in (\ref{equ:q_function}) and substituting in (\ref{equ:BER_noisy}), one obtains the closed-form expression in (\ref{equ:BER_noisy2}), which completes the proof.
\section{Proof of Proposition 2} \label{app:C}
 Using ML detection, it follows that

\begin{equation}\label{equ:outdated1}
\mbox{Pr}_{e_{m\rightarrow k| s_k=0}}= \int_{\frac{1}{2} \gamma P_m \hat{h}_k}^{\infty} \mathcal{N}_{y_k}\left( \gamma h_k \left( \sum_{j=1}^{m-1}  {e_j} P_j + \sum \limits_{i=m+1}^{N}   P_i s_i\right),\sigma_n\right) dy_k.
\end{equation}
The above integral representation can also be expressed in closed-form in terms of the $Q-$function. Based on this, it immediately follows that
\begin{equation}\label{equ:outdated1}
\mbox{Pr}_{e_{m\rightarrow k| s_k=0}}=\mathcal{Q}\left(\frac{1}{\sigma_n}\left(\gamma  \hat{h}_k\frac{P_m}{2}  - \sum_{j=1}^{m-1}  \gamma h_k  {e_j} P_j - \sum \limits_{i=m+1}^{N}  \gamma h_k  P_i s_i\right)\right)
\end{equation}
which can be equivalently expressed as
\begin{equation}\label{equ:outdated1}
\mbox{Pr}_{e_{m\rightarrow k| s_k=0}}= \mathcal{Q}\left(\frac{\gamma P_m}{2 {\sigma_n} }  \mathcal{E}+\frac{ \gamma h_k}{\sigma_n} \left(\frac{P_m}{2} - \sum_{j=1}^{m-1} {e_j} P_j - \sum_{l=m+1}^{N} P_l A_{il}\right)\right).
\end{equation}
Likewise,
\begin{equation}\label{equ:outdated2}
\mbox{Pr}_{e_{m\rightarrow k| s_k=1}}  = \int_{-\infty}^{\frac{1}{2}P_m  \gamma \hat{h_k}} \mathcal{N}_{y_k} \left( \gamma h_k \left(P_m + \sum_{j=1}^{m-1}   {e_j} P_j + \sum \limits_{i=m+1}^{N}  P_i s_i \right)\sigma_n\right) dy_k
\end{equation}
which can be expressed in closed-form as
\begin{equation}\label{equ:outdated2}
\mbox{Pr}_{e_{m\rightarrow k| s_k=1}}=1 - \mathcal{Q}\left(\frac{1}{\sigma_n}\left(\gamma \hat{h}_k\frac{P_m}{2}  -  \gamma h_k P_m - \sum_{j=1}^{m-1}  \gamma h_k  {e_j} P_j - \sum \limits_{i=m+1}^{N}  \gamma h_k  P_i s_i\right)\right)
\end{equation}
and
\begin{equation}\label{equ:outdated2}
\mbox{Pr}_{e_{m\rightarrow k| s_k=1}}=\mathcal{Q}\left(-\frac{\gamma P_m}{2{\sigma_n} }  \mathcal{E} +\frac{ \gamma h_k}{\sigma_n} \left(\frac{P_m}{2} + \sum_{j=1}^{m-1} {e_j} P_j + \sum_{l=m+1}^{N} P_l A_{il}\right)\right)
\end{equation}
%\begin{equation}\label{equ:Q-func2}
%\begin{aligned}
%\begin{cases}
%&\mbox{Pr}_{e_{m\rightarrow k| s_k=0}}\\ &= \frac{1}{\sqrt{2 \pi \sigma_n^2}} \int_{\frac{1}{2}P_m \hat{h}_k}^{\infty} e^{-\frac{\left(y_k - h_k \left( \sum_{j=1}^{m-1}  e_j P_j + \sum \limits_{i=m+1}^{N}   P_i s_i \right) \right)^2}{2 \sigma_n^2}}\\
%&= \mathcal{Q}(\frac{1}{\sigma_n}( \hat{h}_k\frac{P_m}{2}  - \sum_{j=1}^{m-1} h_k  e_j P_j - \sum \limits_{i=m+1}^{N} h_k  P_i s_i))\\
%&= \mathcal{Q}(\frac{P_m}{2 {\sigma_n} }  \mathcal{E}+\frac{h_k}{\sigma_n} (\frac{P_m}{2} - \sum_{j=1}^{m-1} e_j P_j - \sum_{l=m+1}^{N} P_l A_{il}))\\
%&\mbox{Pr}_{e_{m\rightarrow k| s_k=1}} \\& = \frac{1}{\sqrt{2 \pi \sigma_n^2}} \int_{-\infty}^{\frac{1}{2}P_m h_k} e^{-\frac{\left(y_k - h_k \left(P_m + \sum_{j=1}^{m-1}   e_j P_j + \sum \limits_{i=m+1}^{N}  P_i s_i \right) \right)^2}{1 \sigma_n^2}}\\
%&=1 - \mathcal{Q}(\frac{1}{\sigma_n}( \hat{h}_k\frac{P_m}{2}  - h_k P_m - \sum_{j=1}^{m-1} h_k  e_j P_j - \sum \limits_{i=m+1}^{N} h_k  P_i s_i))\\
%&=\mathcal{Q}(-\frac{P_m}{2{\sigma_n} }  \mathcal{E} +\frac{h_k}{\sigma_n} (\frac{P_m}{2} + \sum_{j=1}^{m-1} e_j P_j + \sum_{l=m+1}^{N} P_l A_{il})),
%\end{cases}
%\end{aligned}
%\end{equation}
where ${e_j} = \hat{s_j} - s_j$. Based on this and after carrying out  some algebraic manipulations, equation (\ref{equ:BER3}) is deduced, which completes the proof.
%\end{appendices}

\end{appendices}
\bibliographystyle{IEEEtran}
\bibliography{journal}

\begin{thebibliography}{10}
\providecommand{\url}[1]{#1}
\csname url@samestyle\endcsname
\providecommand{\newblock}{\relax}
\providecommand{\bibinfo}[2]{#2}
\providecommand{\BIBentrySTDinterwordspacing}{\spaceskip=0pt\relax}
\providecommand{\BIBentryALTinterwordstretchfactor}{4}
\providecommand{\BIBentryALTinterwordspacing}{\spaceskip=\fontdimen2\font plus
\BIBentryALTinterwordstretchfactor\fontdimen3\font minus
  \fontdimen4\font\relax}
\providecommand{\BIBforeignlanguage}[2]{{%
\expandafter\ifx\csname l@#1\endcsname\relax
\typeout{** WARNING: IEEEtran.bst: No hyphenation pattern has been}%
\typeout{** loaded for the language `#1'. Using the pattern for}%
\typeout{** the default language instead.}%
\else
\language=\csname l@#1\endcsname
\fi
#2}}
\providecommand{\BIBdecl}{\relax}
\BIBdecl

\bibitem{VLC_mag}
H.~Burchardt, N.~Serafimovski, D.~Tsonev, S.~Videv, and H.~Haas, ``{VLC}:
  Beyond point-to-point communication,'' \emph{IEEE Commun. Mag}, vol.~52,
  no.~7, pp. 98--105, July 2014.

\bibitem{haas_book}
S.~Dimitrov and H.~Haas, \emph{Principles of LED Light Communications: Towards
  Networked Li-Fi}.\hskip 1em plus 0.5em minus 0.4em\relax Cambridge University
  Press, 2015.

\bibitem{VLCmarket}
A.~Jovicic, J.~Li, and T.~Richardson, ``Visible light communication:
  opportunities, challenges and the path to market,'' \emph{IEEE Commun. Mag.},
  vol.~51, no.~12, pp. 26--32, December 2013.

\bibitem{haas3}
H.~Haas, L.~Yin, Y.~Wang, and C.~Chen, ``What is {LiFi?}'' \emph{J. Lightw.
  Technol.}, vol.~34, no.~6, pp. 1533--1544, Mar. 2016.

\bibitem{state-of-the-art}
D.~Karunatilaka, F.~Zafar, V.~Kalavally, and R.~Parthiban, ``{LED} based indoor
  visible light communications: State of the art,'' \emph{Commun. Surveys
  Tuts.}, vol.~17, no.~3, pp. 1649--1678, Aug. 2015.

\bibitem{haas2}
S.~Rajbhandari, H.~Chun, G.~Faulkner, K.~Cameron, A.~V.~N. Jalajakumari,
  R.~Henderson, D.~Tsonev, M.~Ijaz, Z.~Chen, H.~Haas, E.~Xie, J.~J.~D.
  McKendry, J.~Herrnsdorf, E.~Gu, M.~D. Dawson, and D.~O'Brien, ``High-speed
  integrated visible light communication system: Device constraints and design
  considerations,'' \emph{IEEE J. Sel. Areas Commun.}, vol.~33, no.~9, pp.
  1750--1757, Sep. 2015.

\bibitem{murat1}
E.~Sarbazi and M.~Uysal, ``{PHY} layer performance evaluation of the{ IEEE}
  802.15.7 visible light communication standard,'' in \emph{Proc. 2nd
  International Workshop on Optical Wireless Communications {(IWOW)}}, Oct.
  2013, pp. 35--39.

\bibitem{survey2}
P.~H. Pathak, X.~Feng, P.~Hu, and P.~Mohapatra, ``Visible light communication,
  networking, and sensing: A survey, potential and challenges,'' \emph{Commun.
  Surveys Tuts.}, vol.~17, no.~4, pp. 2047--2077, Fourthquarter 2015.

\bibitem{CSMA1}
S.~K. Nobar, K.~A. Mehr, and J.~M. Niya, ``Comprehensive performance analysis
  of {IEEE} 802.15.7 {CSMA/CA} mechanism for saturated traffic,'' \emph{IEEE J.
  Opt. Commun. Netw.}, vol.~7, no.~2, pp. 62--73, Feb. 2015.

\bibitem{CSMA_FD}
Z.~Zhang, X.~Yu, L.~Wu, J.~Dang, and V.~O.~K. Li, ``Performance analysis of
  full-duplex visible light communication networks,'' in \emph{Proc. IEEE
  International Conference on Communications {(ICC)}}, June 2015, pp.
  3933--3938.

\bibitem{csma2}
L.~Zhao, X.~Chi, and W.~Shi, ``A {QoS}-driven random access algorithm for
  {MPR}-capable {VLC} system,'' \emph{IEEE Commun. Lett}, vol.~20, no.~6, pp.
  1239--1242, June 2016.

\bibitem{OFDMA1}
J.~Dang and Z.~Zhang, ``Comparison of optical {OFDM-IDMA} and optical {OFDMA}
  for uplink visible light communications,'' in \emph{Proc. International
  Conference on Wireless Communications Signal Processing (WCSP)}, Oct. 2012,
  pp. 1--6.

\bibitem{OOC1}
J.~A. Salehi, ``Code division multiple-access techniques in optical fiber
  networks. {I}. fundamental principles,'' \emph{IEEE Trans. Commun.}, vol.~37,
  no.~8, pp. 824--833, Aug. 1989.

\bibitem{OOC2}
J.~A. Salehi and C.~A. Brackett, ``Code division multiple-access techniques in
  optical fiber networks. {II}. systems performance analysis,'' \emph{IEEE
  Trans. Commun.}, vol.~37, no.~8, pp. 834--842, Aug. 1989.

\bibitem{CDMA1}
M.~Guerra-Medina, B.~Rojas-Guillama, O.~Gonzalez, J.~Martin-Gonzalez, E.~Poves,
  and F.~Lopez-Hernandez, ``Experimental optical code-division multiple access
  system for visible light communications,'' in \emph{Proc. Wireless
  Telecommunications Symposium (WTS)}, Apr. 2011, pp. 1--6.

\bibitem{cdma4}
P.~V. Kumar, R.~Omrani, J.~Touch, A.~E. Willner, and P.~Saghari, ``{CTH01-5}: A
  novel optical {CDMA} modulation scheme: Code cycle modulation,'' in
  \emph{Proc. {IEEE} Globecom}, Nov. 2006, pp. 1--5.

\bibitem{CDMA2}
M.~F. Guerra-Medina, O.~Gonzalez, B.~Rojas-Guillama, J.~A. Martin-Gonzalez,
  F.~Delgado, and J.~Rabadan, ``Ethernet-{OCDMA} system for multi-user visible
  light communications,'' \emph{Electronics Letters}, vol.~48, no.~4, pp.
  227--228, Feb. 2012.

\bibitem{NOMA_VLC}
H.~Marshoud, V.~M. Kapinas, G.~K. Karagiannidis, and S.~Muhaidat,
  ``Non-orthogonal multiple access for visible light communications,''
  \emph{{IEEE} Photon. Technol. Lett.}, vol.~28, no.~1, pp. 51--54, Jan. 2016.

\bibitem{whiltepaper}
{NTT DOCOMO Inc., Tokyo, Japan}, ``{5G} radio access: requirements, concept and
  technologies,'' \emph{{5G} Whitepaper}.

\bibitem{NOMA-5g}
L.~Dai, B.~Wang, Y.~Yuan, S.~Han, C.~l.~I, and Z.~Wang, ``Non-orthogonal
  multiple access for {5G}: solutions, challenges, opportunities, and future
  research trends,'' \emph{IEEE Commun. Mag.}, vol.~53, no.~9, pp. 74--81, Sep.
  2015.

\bibitem{noma1}
Z.~Ding, Z.~Yang, P.~Fan, and H.~Poor, ``On the performance of non-orthogonal
  multiple access in $5${G} systems with randomly deployed users,''
  \emph{{IEEE} Signal Process. Lett.}, vol.~21, no.~12, pp. 1501--1505, Dec.
  2014.

\bibitem{NOMA-relay}
J.~Men and J.~Ge, ``Non-orthogonal multiple access for multiple-antenna
  relaying networks,'' \emph{IEEE Commun. Lett}, vol.~19, no.~10, pp.
  1686--1689, Oct. 2015.

\bibitem{NOMA-MIMO1}
Z.~Ding, F.~Adachi, and H.~V. Poor, ``The application of {MIMO} to
  non-orthogonal multiple access,'' \emph{IEEE Trans. Wireless Commun.},
  vol.~15, no.~1, pp. 537--552, Jan. 2016.

\bibitem{NOMA-MIMO2}
Z.~Ding and H.~V. Poor, ``Design of massive-{MIMO-NOMA} with limited
  feedback,'' \emph{{IEEE} Signal Process. Lett.}, vol.~23, no.~5, pp.
  629--633, May 2016.

\bibitem{NOMA-SCI1}
Q.~Sun, S.~Han, C.~L. I, and Z.~Pan, ``On the ergodic capacity of {MIMO NOMA}
  systems,'' \emph{{IEEE} Wireless Communications Letters}, vol.~4, no.~4, pp.
  405--408, Aug. 2015.

\bibitem{NOMA-CSI2}
Z.~Yang, Z.~Ding, P.~Fan, and G.~K. Karagiannidis, ``On the performance of
  non-orthogonal multiple access systems with partial channel information,''
  \emph{IEEE Trans. Commun.}, vol.~64, no.~2, pp. 654--667, Feb. 2016.

\bibitem{NOMA_VLC1}
L.~Yin, X.~Wu, and H.~Haas, ``On the performance of non-orthogonal multiple
  access in visible light communication,'' in \emph{Proc. IEEE 26th Annual
  International Symposium on Personal, Indoor, and Mobile Radio Communications
  {(PIMRC)}}, Aug. 2015, pp. 1354--1359.

\bibitem{NOMA_haas}
L.~Yin, W.~O. Popoola, X.~Wu, and H.~Haas, ``Performance evaluation of
  non-orthogonal multiple access in visible light communication,'' \emph{IEEE
  Trans. Commun.}, vol.~PP, no.~99, pp. 1--1, 2016.

\bibitem{NOMA_VLC2}
R.~C. Kizilirmak, C.~R. Rowell, and M.~Uysal, ``Non-orthogonal multiple access
  {(NOMA)} for indoor visible light communications,'' in \emph{Proc. 4th
  International Workshop on Optical Wireless Communications {(IWOW)}}, Sep.
  2015, pp. 98--101.

\bibitem{Joint}
K.~Ying, H.~Qian, R.~Baxley, and S.~Yao, ``Joint optimization of precoder and
  equalizer in {MIMO} {VLC} systems,'' \emph{IEEE J. Sel. Areas Commun},
  vol.~33, no.~9, pp. 1949--1958, Sep. 2015.

\bibitem{marshoud}
H.~Marshoud, D.~Dawoud, V.~Kapinas, G.~K.~Karagiannidis, S.~Muhaidat, and
  B.~Sharif, ``{MU-MIMO} precoding for {VLC} with imperfect {CSI},'' in
  \emph{Proc. 4th International Workshop on Optical Wireless Communications
  {(IWOW)}}, Sep. 2015, pp. 93--97.

\bibitem{coordinated}
H.~Ma, L.~Lampe, and S.~Hranilovic, ``Coordinated broadcasting for multiuser
  indoor visible light communication systems,'' \emph{IEEE Trans. Commun.},
  vol.~63, no.~9, pp. 3313--3324, Sep. 2015.

\bibitem{dimming4}
J.~Gancarz, H.~Elgala, and T.~D.~C. Little, ``Impact of lighting requirements
  on {VLC} systems,'' \emph{IEEE Commun. Mag}, vol.~51, no.~12, pp. 34--41,
  Dec. 2013.

\bibitem{ook_standard}
``{IEEE }standard for local and metropolitan area networks--part 15.7:
  Short-range wireless optical communication using visible light,'' \emph{in
  IEEE Std 802.15.7-2011}, pp. 1--309, Sep. 6 2011.

\bibitem{dimming3}
K.~Lee and H.~Park, ``Modulations for visible light communications with dimming
  control,'' \emph{IEEE Photon. Technol. Lett.}, vol.~23, no.~16, pp.
  1136--1138, Aug. 2011.

\bibitem{dimming5}
I.~Stefan, H.~Elgala, and H.~Haas, ``Study of dimming and {LED} nonlinearity
  for {ACO-OFDM} based {VLC} systems,'' in \emph{Proc. IEEE Wireless
  Communications and Networking Conference (WCNC)}, Apr. 2012, pp. 990--994.

\bibitem{murat2}
F.~Miramirkhani and M.~Uysal, ``Channel modeling and characterization for
  visible light communications,'' \emph{IEEE Photon. J.}, vol.~7, no.~6, pp.
  1--16, Dec. 2015.

\bibitem{modeling_LED}
\BIBentryALTinterwordspacing
I.~Moreno and C.-C. Sun, ``Modeling the radiation pattern of {LEDs},''
  \emph{Opt. Express}, vol.~16, no.~3, pp. 1808--1819, Feb. 2008. [Online].
  Available: \url{http://www.opticsexpress.org/abstract.cfm?URI=oe-16-3-1808}
\BIBentrySTDinterwordspacing

\bibitem{Fundamental}
T.~Komine and M.~Nakagawa, ``Fundamental analysis for visible-light
  communication system using {LED} lights,'' \emph{IEEE Trans. Consum.
  Electron.}, vol.~50, no.~1, pp. 100--107, Feb. 2004.

\bibitem{quantization}
C.~C. Trinca, J.~C. Belfiore, E.~D. de~Carvalho, and J.~V. Filho, ``Estimation
  with mean square error for real-valued channel quantization,'' in \emph{Proc.
  {IEEE} Globecom Workshops ({GC} Wkshps)}, Dec. 2014, pp. 275--280.

\bibitem{dimming1}
F.~Zafar, D.~Karunatilaka, and R.~Parthiban, ``Dimming schemes for visible
  light communication: the state of research,'' \emph{IEEE Trans. Wireless
  Commun.}, vol.~22, no.~2, pp. 29--35, Apr. 2015.

\bibitem{dimming2}
S.~H. Lee, S.~Y. Jung, and J.~K. Kwon, ``Modulation and coding for dimmable
  visible light communication,'' \emph{IEEE Commun. Mag}, vol.~53, no.~2, pp.
  136--143, Feb. 2015.

\bibitem{CSI1}
M.~Seyfi, S.~Muhaidat, and J.~Liang, ``Amplify-and-forward selection
  cooperation over rayleigh fading channels with imperfect {CSI},'' \emph{IEEE
  Trans. Wireless Commun.}, vol.~11, no.~1, pp. 199--209, Jan. 2012.

\bibitem{CSI2}
M.~J. Taghiyar, S.~Muhaidat, and J.~Liang, ``On the performance of pilot symbol
  assisted modulation for cooperative systems with imperfect channel
  estimation,'' in \emph{Proc. {IEEE} Wireless Communication and Networking
  Conference {(WCNC)}}, Apr. 2010, pp. 1--5.

\bibitem{group_mobility}
J.~Liu, N.~Kato, J.~Ma, and T.~Sakano, ``Throughput and delay tradeoffs for
  mobile ad hoc networks with reference point group mobility,'' \emph{IEEE
  Trans. Wireless Commun.}, vol.~14, no.~3, pp. 1266--1279, Mar. 2015.

\bibitem{Q-function}
M.~Lopez-Benitez and F.~Casadevall, ``Versatile, accurate, and analytically
  tractable approximation for the {Gaussian} {Q}-function,'' \emph{IEEE Trans.
  Commun.}, vol.~59, no.~4, pp. 917--922, April 2011.

\end{thebibliography}

\end{document}